\documentclass[12pt]{iopart}
\usepackage{iopams}
\expandafter\let\csname equation*\endcsname\relax

\expandafter\let\csname endequation*\endcsname\relax

\usepackage{cite}
\usepackage{amsmath}
\usepackage{amstext,amsthm}
\usepackage{esint}
\usepackage{graphicx}
\usepackage[colorlinks,linkcolor=black,urlcolor=blue,citecolor=blue]{hyperref}
\usepackage{bm,bbm}
\usepackage[utf8]{inputenc} 
\usepackage{amsfonts} 
\usepackage{nicefrac}
\usepackage{bbm}
\usepackage{mathrsfs}
\usepackage{pgfplots}
\usepgfplotslibrary{fillbetween}
\usepackage{physics}
\usepackage{dsfont}
\usepackage{graphicx}
\usepackage{subcaption}

\theoremstyle{plain}
\newtheorem{thm}{Theorem}
\newtheorem{conj}{Conjecture}
\newtheorem{prop}{Proposition}

\newtheorem{mydef}{Definition}

\newtheorem{quest}{Question}
\newenvironment{questbis}[1]
  {\renewcommand{\thequest}{\ref{#1}$'$}%
   \addtocounter{quest}{-1}%
   \begin{quest}}
  {\end{quest}}
\theoremstyle{remark}

\newcommand{\Hi}{\mathscr{H}}

\newcommand{\numberset}{\mathbb}

\newcommand{\R}{\numberset{R}} 
\newcommand{\C}{\numberset{C}}
\newcommand{\ud}{\,\mathrm{d}}
\newcommand{\LOCC}{\operatorname{LOCC}}

\newcommand{\be}{\begin{equation}}
\newcommand{\ee}{\end{equation}}

\DeclareMathAlphabet{\mathcalligra}{T1}{calligra}{m}{n}

\begin{document}
\title{Volume of the set of LOCC-convertible quantum states}

\author{Fabio Deelan Cunden$^{1}$, Paolo Facchi$^{2,3}$, Giuseppe Florio$^{3,4}$, Giovanni Gramegna$^{2,3}$}
\address{$^{1}$ School of Mathematics and Statistics, University College Dublin, Dublin 4, Ireland }
\address{$^{2}$ Dipartimento di Fisica and MECENAS, Università di Bari, I-70126 Bari, Italy }
\address{$^{3}$ INFN, Sezione di Bari, I-70126 Bari, Italy}
\address{$^{4}$ Dipartimento di Meccanica, Matematica e Management, Politecnico di Bari, Via E. Orabona 4, I–70125 Bari, Italy}

    \begin{abstract}
The class of quantum operations known as Local Operations and Classical Communication (LOCC) induces a partial ordering on quantum states. We present the results of systematic numerical computations related to the volume (with respect to the unitarily invariant measure) of the set of LOCC-convertible bipartite pure states, where the ordering is characterised by an algebraic relation known as majorization.
The numerical results, which exploit a tridiagonal model of random matrices, provide quantitative evidence that the proportion of LOCC-convertible pairs vanishes in the limit of large dimension, and therefore support a previous conjecture by Nielsen. In particular, we show that the problem is equivalent to the persistence of a non-Markovian stochastic process and the proportion of LOCC-convertible pairs decays algebraically with a nontrivial persistence exponent. We extend this analysis by investigating the distribution of the maximal success probability of LOCC-conversions. We show a dichotomy in behaviour between balanced and unbalanced bipartitions. In the 
latter case the asymptotics is somehow surprising: in the limit of large dimensions, for the overwhelming majority of pairs of states a perfect LOCC-conversion is not possible; nevertheless, for most states there exist local strategies that succeed in achieving the conversion with a probability arbitrarily close to one. We present strong evidences of a universal scaling limit for the maximal probability of successful LOCC-conversions and we suggest a connection with the typical fluctuations of the smallest eigenvalue of Wishart random matrices. 
    \end{abstract}
\maketitle
%
\section{Introduction and summary of results}

The theoretical and experimental development of quantum information protocols has motivated the study of `manipulations' of quantum states using classes of operations restricted by some physical constraints. 
When considering systems composed of two parties $A$ and $B$ (or more), the physical constraint corresponds to the paradigm of `distant labs': $A$ and $B$ can only perform local unitaries and measurements, and can freely communicate classical data, which includes the results of their local measurements. The subset of quantum operations describing this scenario is the so-called class of \emph{Local Operations assisted by Classical Communication} (LOCC). In the language of resource theory, LOCC correspond to those operations which can be implemented without consuming entanglement.   
Despite having a quite simple physical description, the mathematical characterization of LOCC is notoriously challenging~\cite{Bennet99,DHR02}, due to the fact that the particular choice of local measurements to perform at some stage of a LOCC protocol depends on all the measurement results previously obtained (and shared between the parties).
For a precise mathematical description, see~\cite{Chitambar2014}. 

We use the following convention throughout the paper.  The density matrix $\ketbra{\psi}{\psi}$ of the pure state $\ket{\psi}$, with $\bra{\psi}\psi\rangle=1$, will be frequently written simply as $\psi$, namely
\begin{equation}
\psi = \ketbra{\psi}{\psi}.
\end{equation}
The reduced density matrix  
of a bipartite state will be denoted by 
\begin{equation}
\psi_A = \Tr_B(\psi),
\end{equation}
and we write $\lambda(\psi_A)$ to denote the spectrum of $\psi_A$, i.e.\ the set of eigenvalues of $\psi_A$.

Consider a system composed of two parties $A$ and $B$. We say that a state $\ket{\psi}$ can be locally converted (or LOCC-converted) into a state $\ket{\varphi}$ if there is a LOCC operation that maps $\psi=\ketbra{\psi}{\psi}$ into $\varphi=\ketbra{\varphi}{\varphi}$. 
Recall that in the bipartite case such a transformation can always be completed in one round of LOCC if it is possible at all~\cite{Lo2001}. See~\ref{appLOCC}.
Here is a generalisation to include local conversions that succeed with some probability $p$ at transforming the initial state into the wished final state. 
\begin{mydef}\label{def:LOCCp}
	Let $0\leq p \leq 1$. We say that $\ket{\psi}$ can be locally converted into $\ket{\varphi}$ with probability of success $p$, and write $\ket{\psi}\xrightarrow{p}\ket{\varphi}$, if there exists a LOCC protocol which yields the state $\varphi$ with probability $p$ when performed on the starting state $\psi$. For a given pair $\ket{\psi}$ and $\ket{\varphi}$, we denote by 
	\begin{equation}
	\Pi(\psi\rightarrow \varphi)
	\end{equation}
	the maximal  $p$ such that $\ket{\psi}\xrightarrow{p} \ket{\varphi}$. We will write $\ket{\psi}\rightarrow\ket{\varphi}$ to say that the conversion can be carried out deterministically, i.e.\  $\Pi(\psi\rightarrow \varphi) =1$.
\end{mydef}
For the readers who are not familiar with these notions, we provide some details on the probability of success associated to a LOCC protocol in~\ref{appLOCC}.

\par
In the present exploratory numerical work we intend to discuss a series of basic questions related to the `volume' of the set of LOCC-convertible quantum states.

\begin{quest}\label{q:1} Pick two states $\ket{\psi}$ and $\ket{\varphi}$ at random. What is the probability that $\ket{\psi}\rightarrow\ket{\varphi}$  (i.e.\ that $\Pi\left(\psi \rightarrow \varphi \right)=1$)?

\end{quest}

If the probability measure is unitarily invariant, the question can be rephrased in terms of volumes of the set of convertible states.
 \par
Nielsen conjectured  that, \emph{`in the limit where $A$ and $B$ are of large dimensionality, almost all pairs of pure states picked according to the unitarily invariant measure of $AB$ will be incomparable'}~\cite{nielsen1999}. 
Nielsen's conjecture is still open.
\par

\par
Questions about the volume of LOCC-convertible states can be relaxed to include local convertibility that succeed with some probability $p$ according to Definition~\ref{def:LOCCp}.

\begin{quest}\label{q:2} Pick two states $\ket{\psi}$ and $\ket{\varphi}$  at random. Given $p\in[0,1]$, what is the probability that $\ket{\psi}\xrightarrow{p}\ket{\varphi}$ (i.e.\ that $\Pi\left(\psi \rightarrow \varphi \right)\geq p$)?
\end{quest}

Question~\ref{q:1} and Question~\ref{q:2} can be rephrased more generally as follows.

\begin{quest}\label{q:3} Pick two states $\ket{\psi}$ and $\ket{\varphi}$ at random.
Find the probability distribution of $\Pi\left(\psi \rightarrow \varphi \right)$. 
\end{quest}

In this work we address the above three questions.
The investigation of local convertibility of quantum states has led to precise criteria based on the \emph{majorization}~\cite{nielsen1999,vidal1999,vidal2000entanglement,nielsen2001majorization}, a relation in the technical mathematical sense that formalizes the idea that a probability distribution can be more `disordered' than another~\cite{MOAinequalities,Bhatia}.  
Remarkably, the linear-algebraic theory of majorization appears to play a key role in other resource theories, such as in the theory of noisy operations~\cite{horodecki2003reversible,gour2015resource}, of total quantum coherence~\cite{yang2018operational}, and of thermal operations~\cite{horodecki2013fundamental,brandao2013resource}. 
Hereafter we let $\ket{\psi}$ and $\ket{\varphi}$ be independent random pure states distributed according to the unitarily invariant measure on the unit sphere of $\Hi=\Hi_A\otimes\Hi_B$, where $\Hi_A=\C^n$ and $\Hi_B=\C^m$.  It is known that, since only local unitary transformations are allowed on $A$ and $B$, and these are locally reversible, any two states with the same Schmidt coefficients are locally equivalent. Thus only the Schmidt coefficients of $\ket{\psi}$ and $\ket{\varphi}$ are relevant as far as local operations are concerned. 
To be more precise, let $\lambda(\psi_A)$ and $\lambda(\varphi_A)$ be the \textit{entanglement spectra} of $\ket{\psi}$ and $\ket{\varphi}$, i.e.\ the $n$-tuples of eigenvalues (squares of the Schmidt coefficients) of the reduced density matrices of $\ket{\psi}$ and $\ket{\varphi}$, respectively. Vidal~\cite{vidal1999} discovered that the maximal probability of local conversion depends on the entanglement spectra only and can be written as 
\begin{equation}
\Pi\left(\psi \rightarrow \varphi \right)=\Pi(\lambda(\psi_A),\lambda(\varphi_A)),
\end{equation} 
where $\Pi(\cdot,\cdot)$ is an explicit function defined, for any pair of probability vectors $x$ and $y$ with $n$ nonzero components, as

\be
\Pi(x,y)=\min_{1\leq k\leq n}\left\{\frac{x^\downarrow_k+x^\downarrow_{k+1}+\cdots +x^\downarrow_n}{y^\downarrow_k+y^\downarrow_{k+1}+\cdots +y^\downarrow_n}\right\},
\ee 
where $x^{\downarrow}$ and $y^{\downarrow}$ are the decreasing rearrangements of $x$ and $y$, respectively. The properties of $\Pi(x,y)$ and its algebraic significance will be discussed later.

We can, therefore, give an equivalent formulation of Question~\ref{q:3} as follows.
\begin{questbis}{q:3} 
\label{q:3b} Pick two states $\ket{\psi}$ and $\ket{\varphi}$ at random.
Find the probability distribution of $\Pi(\lambda(\psi_A),\lambda(\varphi_A))$.
\end{questbis}
We also remark at this stage that the Schmidt decomposition is \emph{symmetric} under the interchange of $A$ and $B$; hence, without loss of generality we can assume $n\leq m$.

\subsection{Numerical Methods}

An ordinary way to estimate the probability of getting $\Pi\left(\psi \rightarrow \varphi \right)=1$ or, more generally, the distribution of $\Pi\left(\psi \rightarrow \varphi \right)$, would be to sample independent pairs of states $\ket{\psi},\ket{\varphi}$, compute their local spectra $\lambda(\psi_A)$ and $\lambda(\varphi_A)$, and then compute $\Pi(\lambda(\psi_A),\lambda(\varphi_A))$.
 This method is reasonable in the estimation problem for small dimensions $n$ and $m$. 
One of the bottle-necks of the procedure is that the dimension of the total Hilbert space $\Hi$ is $nm$: this is a severe limitation when sampling the  states $\ket{\psi}$ and $\ket{\varphi}$ if $n$ or $m$ are large (the asymptotic regime of interest here).  

This problem can be bypassed by noticing that 
the local spectra $\lambda(\psi_A)$ and $\lambda(\varphi_A)$ are independent and distributed according to a close relative of the celebrated Laguerre unitary ensemble (eigenvalues of complex Wishart matrices of size $n$ and parameter $m$) well studied in the theory of random matrices~\cite{mehta2004random,Vivo}. A notable result by Dumitriu and Edelman~\cite{edel2002mat} implies that \emph{for any $m$}, the Laguerre unitary ensemble can be generated using a tridiagonal method that only requires $\Or(n)$ random real numbers. Using a simple adaptation of Dumitriu and Edelman tridiagonal method, we sampled independent pairs of local spectra $\lambda(\psi_A),\lambda(\varphi_A)$ (without sampling $\ket{\psi}$ and $\ket{\varphi}$!), and then computed 
$\Pi(\lambda(\psi_A),\lambda(\varphi_A))$. This shortcut made possible numerical calculations for dimensions $n$ and $m=c n$ with $n=2,\dots,1024$, and $c=1,\dots,10^3$. Simulations for such large values of $n$ and $m$ by using the ordinary method are impractical. 
\par

\subsection{Summary of Results}
Our findings are the following.
\begin{enumerate}
\item For fixed $n$, numerical computations show that the sequence $P\left(\ket{\psi}\rightarrow\ket{\varphi}\right)$ is non-decreasing in $m$. Therefore, as $m\to\infty$, the probability of local conversion converges to a positive constant:
\begin{equation}
\lim_{m\to\infty}P\left(\ket{\psi}\rightarrow\ket{\varphi}\right)=\kappa(n)>0.
\end{equation}
The sequence of constants $\kappa(n)$ is decreasing in $n$. See \figurename{ \ref  {fig:Nielsen}(a)}.
The result is of course symmetric under exchange of $n$ and $m$.
\item Suppose that both $n,m\to \infty$. We verified numerically that  the relative volume of pairs of LOCC-convertible states goes to zero
\begin{equation}
\lim_{n,m\to\infty}P\left(\ket{\psi}\rightarrow\ket{\varphi}\right)=0.
\end{equation}
(This is a precise statement of Nielsen's conjecture.) See \figurename{ \ref  {fig:Nielsen}(b)}.

\begin{figure}[t]
	\centering
	\begin{subfigure}[t]{0.49\textwidth}
		\centering
		\includegraphics[width=\textwidth]{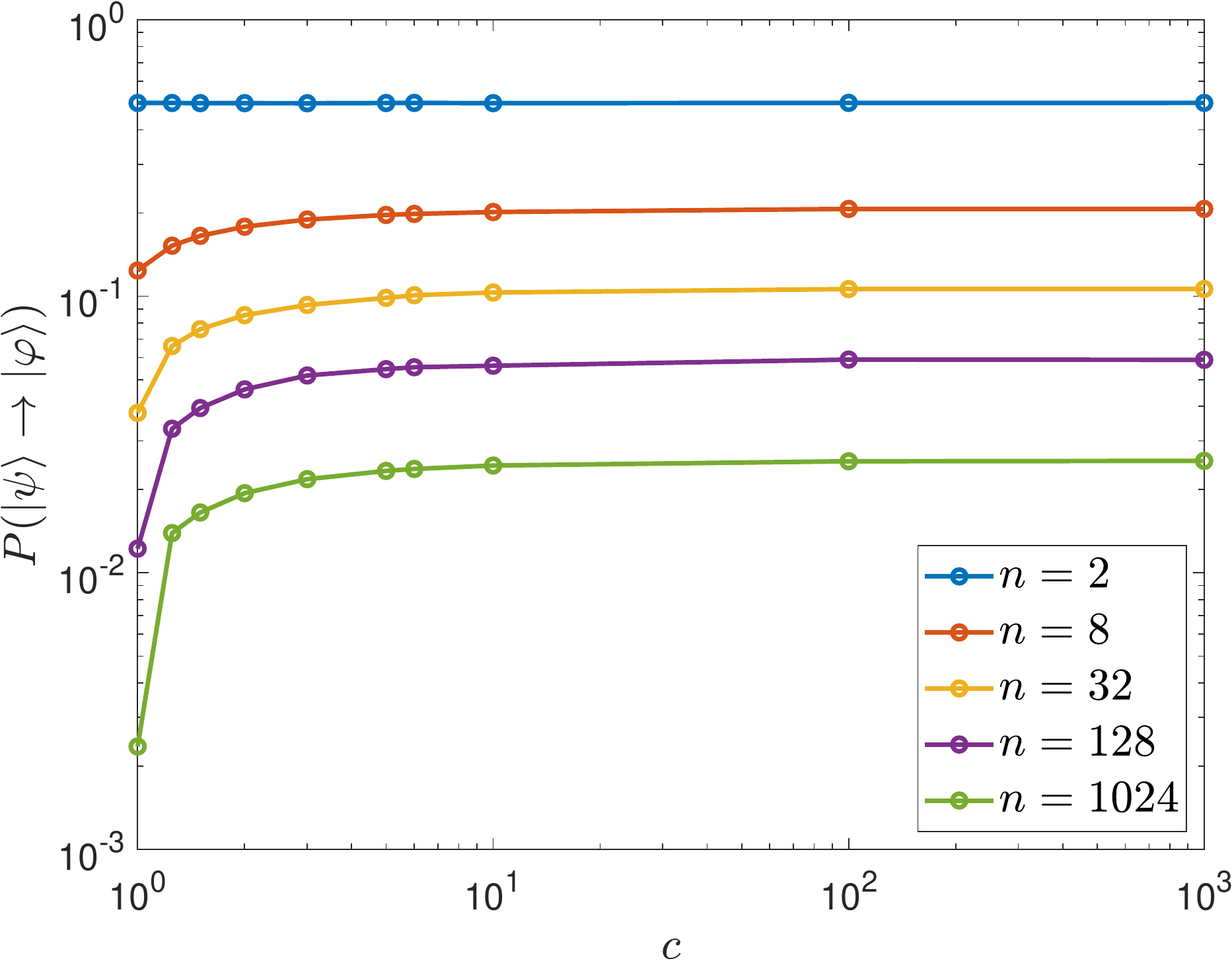}
		\caption{}
		\label{fig:NielsenUnb}
	\end{subfigure}~
	\begin{subfigure}[t]{0.49\textwidth}
		\centering
		\includegraphics[width=\textwidth]{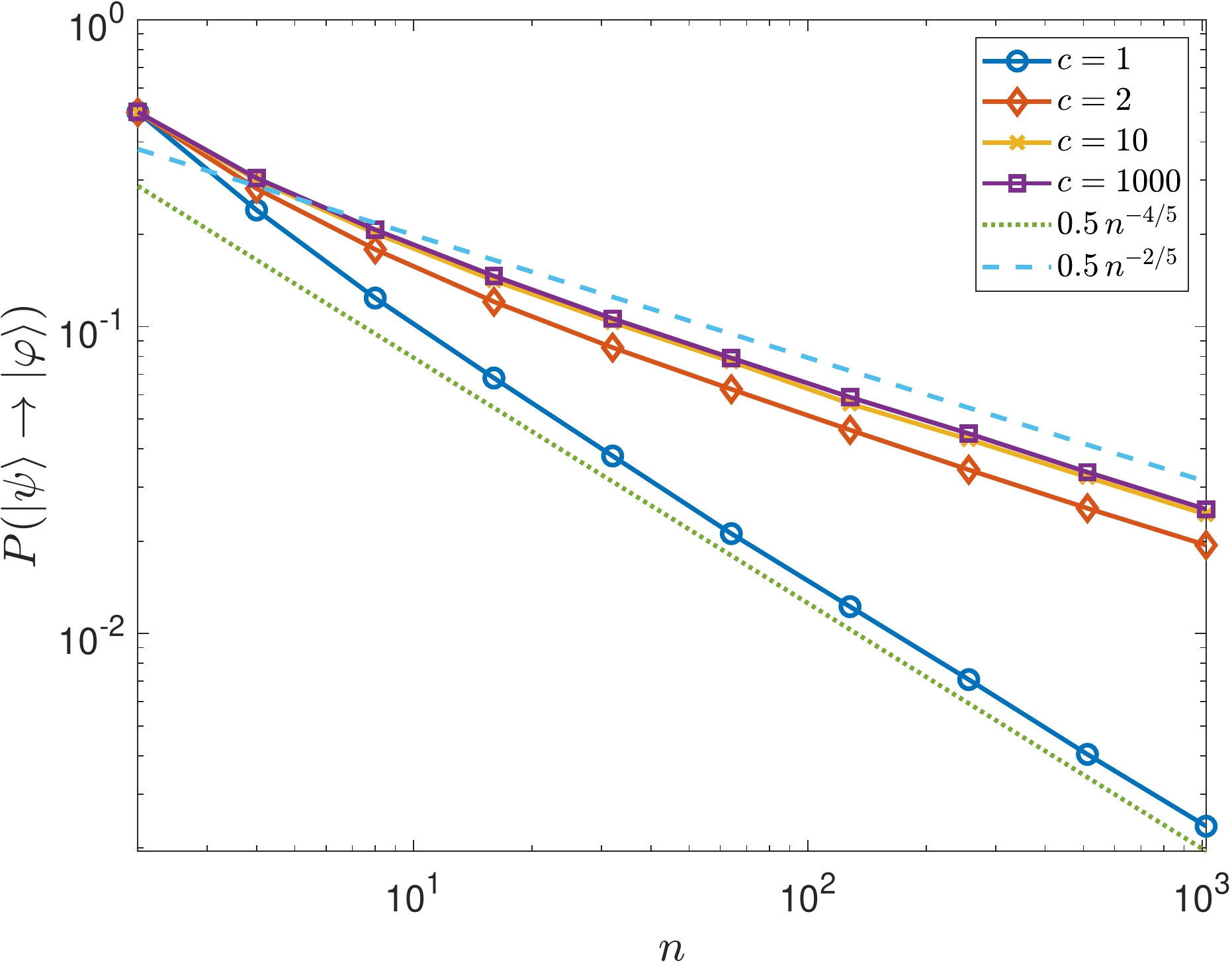}
		\caption{}
		\label{fig:NielsenDim}
	\end{subfigure}
	\caption{The probability $P\left(\ket{\psi}\rightarrow\ket{\varphi}\right)$ for random pure states $\psi$ and $\varphi$  versus the parameter $c$ for several values of $n$ (a), and versus the dimension $n$ for several values of the parameter $c=m/n$ (b).}
	\label{fig:Nielsen}
\end{figure}%
\item We can analyse the asymptotic rate at which the probability of local conversion decays to zero.
More precisely, we consider the limit $n,m\to \infty$ with the ratio $\frac{m}{n}=c$ fixed. We can assume $c\geq1$. In fact, the results are symmetric with respect to the exchange of $n$ and $m$, i.e.\ the transformation $c\leftrightarrow{c}^{-1}$ (the fixed point $c=1$ of the transformation is indeed quite special as discussed later). The analysis of the problem reveals a mapping to the calculation of the persistence probability  for a non-Markovian random walk (see Section~\ref{sec:incomparability}). This connection, combined with numerical evidence (see \figurename{ ~\ref{fig:Nielsen}}), suggests that the probability of local conversion decays \emph{algebraically} (as a power law) at large $n$. Indeed, we find numerically that
\begin{equation}
P\left(\ket{\psi}\rightarrow\ket{\varphi}\right)=P\left(\Pi\left(\psi \rightarrow \varphi \right)=1\right)\simeq \frac{b}{n^{\theta}},\quad\text{as $n\to\infty$ with $m=cn$}.
\end{equation}
A few values of the so-called `persistence exponent' $\theta$  and prefactor $b$ are shown in \tablename{ \ref  {tab:table1}}.

  \begin{table}[h!]
  \begin{center}
    \setlength{\tabcolsep}{1.2em} 
{\renewcommand{\arraystretch}{1.3}
    \begin{tabular}{ccc} 
      $c$ & $\theta$ & $b$\\
       \hline   \hline
		$1$		& $0.795\pm0.013$ 	& $0.579\pm 0.041$\\
   		$3/2$	& $0.420\pm 0.006$ 	& $0.302\pm 0.010$\\
   		$2$ 	& $0.418\pm0.006$ 	& $0.348\pm 0.010$\\
      	$5$ 	& $0.409\pm0.005$ 	& $0.394\pm 0.011$\\
  		$10$ 	& $0.400\pm0.005$ 	& $0.391\pm 0.010$\\
   		$100$ 	& $0.408\pm0.005$ 	& $0.423\pm 0.011$\\
   		$1000$ 	& $0.405\pm0.005$ 	& $0.422\pm 0.011$
    \end{tabular}
    }
  \end{center}
  
\caption{Parameters $b$ and $\theta$ of the best fitting curve for the probability of local conversion, $P(\ket{\psi}\rightarrow \ket{\varphi})\simeq b/n^{\theta}$. The fit is taken only over the last 4 points of each curve in \figurename{ \ref  {fig:Nielsen}(b)}.}
\label{tab:table1}
\end{table} 

Hence, we have that the persistence exponent is $\theta\simeq 4/5$ if $c=1$, and $\theta\simeq 2/5$ if $c>1$. 
\item We computed numerically the probability distribution 
\begin{equation}\label{eq:F(p)}
F(p)=P\left(\Pi\left(\psi \rightarrow \varphi \right)\leq p\right)
\end{equation}
for several values of $n$ and $m$.
It turns out that the probability density $f(p)=F'(p)$ of $\Pi\left(\psi \rightarrow \varphi \right)$ has  a continuous part (supported on the interval $0\leq p \leq1$) and a singular part at $p=1$: the non-decreasing function $F(p)$ has a jump at $p=1$. The height of the jump $F(1)-F(1^-)=P\left(\Pi\left(\psi \rightarrow \varphi \right)=1\right)$ vanishes in the limit $n\to\infty$ with $m=cn$ (see item (ii)). Amusingly, while the weight of singular part at $p=1$ vanishes for large $n$, the continuous part of the density concentrates around (on the left of) $p=1$. In formulae, we found evidence that for $m=cn$ ($c>1$ fixed),
\begin{equation}
\lim_{n\to\infty}F(p)=\theta(p-1),
\label{eq:conj_F}
\end{equation}
{where $\theta(x)$ is the Heaviside step function, defined as $\theta(x)=1$ for $x\geq 0$ and zero otherwise.	
}
\item We can substantiate the previous point by an exact calculation for $n=2$ and generic $m\geq2$.
We have 
\begin{equation}
f(p)= f_{\mathrm{cont}}(p)+\frac{1}{2}\delta(1-p),
\label{eq:n=2Vid}
\end{equation}
with
\begin{equation}
f_{\mathrm{cont}}(p)=\int_{0}^{1/2}xq_{2,m}(x)q_{2,m}(px)\ud x,\quad q_{2,m}(x)=\frac{\Gamma(2m)}{\Gamma(m)\Gamma(m-1)}(x-x^2)^{m-2}(1-2x)^2,
\label{eq:n=2Vid2}
\end{equation}
where $\Gamma(z)=\int_0^\infty t^{z-1}e^{-t}\ud t$ is the Euler $\Gamma$ function. 
Of course, $f_{\mathrm{cont}}(p)\geq0$ and $\int_{0}^1 f_{\mathrm{cont}}(p)dp=1/2$. In fact, $ f_{\mathrm{cont}}(p)$ is a polynomial in $p$ of degree $2m-2$. 
By analysing the exact formulae~\eref{eq:n=2Vid}--\eref{eq:n=2Vid2}, we see that $\lim_{m\to\infty}f_{\mathrm{cont}}(p)=0$ if $p<1$ and $\lim_{m\to\infty}f_{\mathrm{cont}}(1)=+\infty$. In fact, it can be shown that $\lim_{m\to\infty}f_{\mathrm{cont}}(p)=\frac{1}{2}\delta(1-p)$. Hence, 
\begin{equation}
\lim_{m\to\infty}F(p)=\theta(p-1),\quad\text{for $n=2$}.
\end{equation} 
Numerical simulations suggest that this is the case for any $n\geq2$ (this statement is stronger than~\eqref{eq:conj_F}).  Again the result is symmetric under exchange of $n$ and $m$. This convergence in distribution would imply convergence in probability 
\be
P\left(\Pi\left(\psi \rightarrow \varphi \right)\leq 1-\epsilon\right)\to0,\quad\text{for every $\epsilon>0$},
\ee
as $m\to\infty$ or $n\to\infty$ (or both with $c>1$).
\item At this point one wishes to understand if a rescaling of the variable $\Pi\left(\psi \rightarrow \varphi \right)$ leads to a nontrivial limit in distribution for large $n$ and $m$. Indeed we found numerically that when $m=cn$ with $c>1$ fixed,
\begin{equation}
\lim_{n\to\infty}P\left(\Pi\left(\psi \rightarrow \varphi \right)> 1-\frac{x}{c^{1/6}\left|1-\sqrt{c}\right|^{2/3}n^{2/3}}\right)=H_{\mathrm{unb}}(x),
\end{equation}
for some nontrivial scaling functions $H_{\mathrm{unb}}(x)$. We can equivalently write for the probability distribution
\begin{equation}
\lim_{n\to\infty}F\left(1-\frac{x}{c^{1/6}\left|1-\sqrt{c}\right|^{2/3}n^{2/3}}\right)=1-H_{\mathrm{unb}}(x).
\label{eq:scaling_unb}
\end{equation}
For $c=1$ (the fixed point of the symmetry $c\leftrightarrow c^{-1}$) we observe instead a simpler asymptotics
\begin{equation}
\lim_{n\to\infty}F(1-x)=1-H_{\mathrm{bal}}(x).
\label{eq:scaling_bal}
\end{equation}
See \figurename{ \ref {fig:VidalRescaledintro}}. Precise expressions for $H_{\mathrm{unb}}(x)$ and $H_{\mathrm{bal}}(x)$ remain beyond reach.

\begin{figure}[t!]
	\centering
		\begin{subfigure}[t]{0.49\textwidth}
		\centering
		\includegraphics[width=\textwidth]{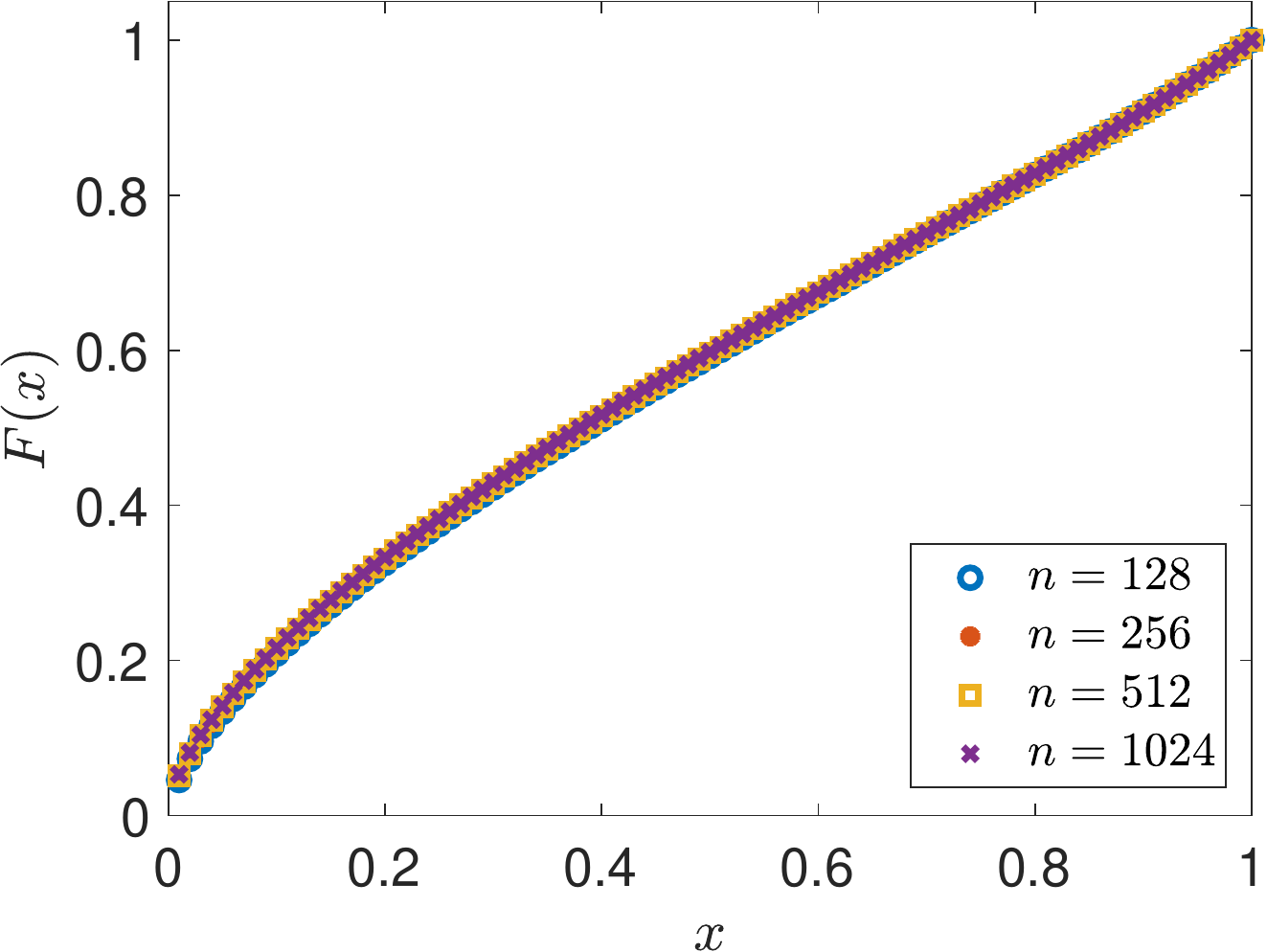}
		\caption{}
		\label{fig:Rescaling_bal}
	\end{subfigure}~
	\begin{subfigure}[t]{.49\textwidth}
		\centering
		\includegraphics[width=\textwidth]{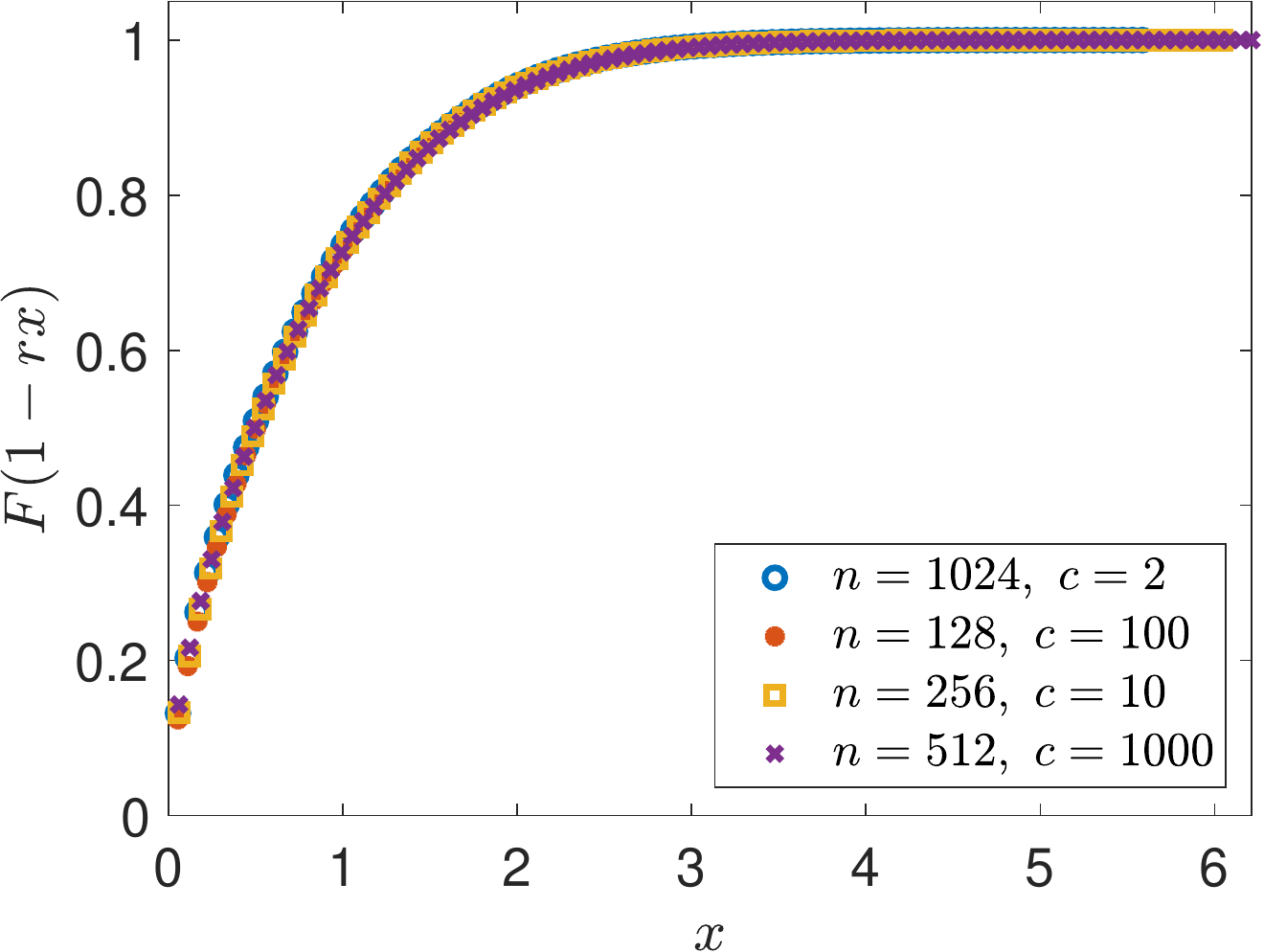}
		\caption{}
		\label{fig:Rescaling_unbal}
	\end{subfigure}\\
	\caption{Universal scaling of the distribution of  $\Pi\left(\psi \rightarrow \varphi \right)$. See Eqs.~\eqref{eq:scaling_unb}--\eqref{eq:scaling_bal}.  Panel (a) shows the balanced case ($c=1$), while panel (b) shows the unbalanced case ($c>1$) and $r=c^{1/6}(\sqrt{c}-1)^{2/3}n^{2/3}$.}
	\label{fig:VidalRescaledintro}
\end{figure}

\end{enumerate}
We stress again that (i)--(vi) above are conjectures in the light of numerical evidence.

The structure of the paper is as follows. In Section~\ref{sec:prelim} we recall the notion of majorization and its connection with entanglement theory through Nielsen's and Vidal's theorems.  
 In Section~\ref{sec:numericalmethod} we compare the standard numerical technique for sampling pure states according to the unitarily invariant measure and a method based on sampling and diagonalising tridiagonal random matrices.
In Section~\ref{sec:incomparability} we review Nielsen's conjecture on the volume of locally convertible states in the asymptotic regime of large dimensions.  In Section~\ref{sec:results} we apply this method to obtain the main results of the paper related to Nielsen's conjecture and the behavior of the maximal conversion probability according to Vidal's theorem. In Section~\ref{sec:conclusion} we draw some conclusions. The precise connection between Wishart matrices and random quantum states is given in~\ref{appA}.

\section{Majorization theory and LOCC-convertibility} \label{sec:prelim}
\subsection{Majorization relation}\label{sub:maj}
For a real vector $x=(x_1,\dots,x_n)$ we denote its decreasing rearrangement by $x^\downarrow$ with components $x_1^\downarrow\geq x_2^\downarrow \geq \dots \geq x_n^\downarrow$.
For two vectors $x,y\in\R^n$, we say that $x$ is \emph{majorized} by $y$---and write $x\prec y$---if 
\begin{align}\label{def:Majorization}
	&\sum_{j=1}^k x_j^\downarrow \leqslant \sum_{j=1}^k y_j^\downarrow,\,\,\text{for $k=1,\dots,n-1$},\qquad  \sum_{j=1}^n x_j^\downarrow =\sum_{j=1}^n y_j^\downarrow.
\end{align}
See~\cite{Bhatia,MOAinequalities}. 
Note that, if $x,y$ are probability vectors, then the last equality is automatically satisfied by virtue of the normalization condition. The majorization relation is a \emph{preorder} on the $(n-1)$-dimensional simplex of probability vectors 
\begin{equation}\label{eq:WeylCh}
\Delta_{n-1}=\bigg\{x\in\R_+^n\colon\sum_{k=1}^n x_k=1\bigg\},
\end{equation}
i.e.\ it is reflexive and transitive.  For $n>2$, the ordering is only partial:  there are vectors $x,y\in\Delta_{n-1}$ such that  neither $x\prec y$ nor $y\prec x$ (sometimes we say that $x$ and $y$ are not comparable). Moreover there are a \emph{smallest} and a \emph{largest} element (up to permutations). Indeed, for every probability vector $x\in \Delta_{n-1}$,
\be
\left(\frac{1}{n},\dots,\frac{1}{n}\right)\prec x\prec (1,0,\dots,0).
\label{eq:sm_larg}
\ee 
Another important property is that
\be
\text{$x\prec z$ and $y\prec z$ } \Rightarrow ax+(1-a)y\prec z,
\ee
for any real $a\in [0,1]$. Hence the set of vectors majorized by $z$ is a convex set, for any $z\in \Delta_{n-1}$.

The majorization condition can be characterised in terms of action of \emph{bistochastic matrices}, i.e.\ matrices with nonnegative entries and whose rows and columns sum to one. (In particular, bistochastic matrices map probability vectors into probability vectors and leave the smallest element $(1/n,\dots,1/n)$ invariant.) The following two classical results  give a geometric insight on the majorization relation. 
\begin{thm}[Hardy-Littlewood-P\'olya~\cite{Hardy52}]\label{thm:HLPthm}
	Let $x,y\in\R^n$. Then, $x\prec y$ if and only if there exists a bistochastic matrix $B$ such that $x=By$.
\end{thm}
\begin{thm}[Birkhoff's Theorem]\label{thm:Birkhoff}
	The set of $n\times n$ bistochastic matrices is a convex set whose extreme points are the permutation matrices.
\end{thm}

Combining the above theorems, we have that $x$ is majorized by $y$ if and only if $x$ is a convex combination of permutations of $y$. 
Informally,  $x\prec y$  if and only if the vector $x$ can be obtained by ``mixing'' the $n!$ reshufflings of $y$ (a natural way to state that $x$ is more disordered than $y$). In fact, the set of probability vectors majorized by $y$ is the convex hull of all vectors that can be obtained by permuting the components of $y$.

\subsection{Local convertibility of quantum states and majorization}\label{sec:conversion}
Consider a bipartite system with finite dimensional Hilbert space $\Hi=\Hi_A\otimes\Hi_B$, with $ \Hi_A=\C^n$, $ \Hi_B=\C^m$, $n\leqslant m$. We also set $c=m/n$.

 Unitary matrices in $\mathrm{U}(nm)$ represent the class of most general transformations between pure states. In fact, the unitary group $\mathrm{U}(nm)$ acts transitively on the unit sphere of $\Hi$: every state $\ket{\psi}$ can be transformed into any wished target state by using a transformation in $\mathrm{U}(nm)$.
 Things change drastically if we consider instead the class of LOCC transformations. (Note that the class LOCC does not have a group structure.) 
 
Recall that a pure state $\ket{\psi}$ is \emph{separable} if $\psi_A$ is a rank-one projection operator; in this case $\lambda(\psi_A)=(1,0,\dots,0)$ (up to a permutation). A state $\ket{\psi}$ is said to be \emph{maximally entangled} if $\psi_A=I/n$; in this case $\lambda(\psi_A)=(1/n,\dots,1/n)$. For $\ket{\psi}\in\Hi$, the number of nonzero components of $\lambda(\psi_A)$ is called Schmidt number of $\ket{\psi}$; if the Schmidt number is strictly larger than one, then the state $\ket{\psi}$ is said to be \emph{entangled}. 
The restriction to LOCC transformations limits the possible state conversions: there are states $\ket{\psi}$, $\ket{\varphi}$ such that a local conversion is not possible, i.e.\ $\Pi\left(\psi \rightarrow \varphi \right)<1$. In particular, `entanglement cannot increase' in a local state conversion. In fact, some pairs of states are \emph{incomparable}, in the sense that neither $\ket{\psi}\rightarrow\ket{\varphi}$ nor $\ket{\varphi}\rightarrow\ket{\psi}$.
\par
For pure states, the `amount of entanglement' is connected in some way to how `uniform' the corresponding vector of Schmidt coefficients is.
When applying a local measurement to a state $\psi$ of $AB$, the order of the eigenvalues $\lambda(\psi_A)$ increases, in accordance with the idea that by means of a measurement we obtain information about the state being measured. This suggests that the majorization relation plays a role in entanglement theory. This intuition turns out to be correct: there is a precise criterion discovered by Nielsen~\cite{nielsen1999} linking the local convertibility of pure states of bipartite systems to the majorization relation between their local spectra. 
\begin{thm}[Nielsen's criterion~\cite{nielsen1999}]
\label{thm:Nielsen}
	Let  $\ket{\psi},\ket{\varphi}\in\Hi$. Then, $\ket{\psi}\rightarrow\ket{\varphi}$ if and only if  $\lambda(\psi_A)\prec\lambda(\varphi_A)$.
\end{thm}
Immediate consequences of Nielsen's criterion are:
\begin{enumerate}
	\item A maximally entangled state can be locally converted into any pure state;
	\item Any pure state can be locally converted into a separable state.
\end{enumerate}
These two examples of LOCC-conversions are always possible because the corresponding local spectra are the smallest and largest elements~\eqref{eq:sm_larg} in the simplex $\Delta_{n-1}$, respectively. On the other hand, as already stressed, the majorization relation is only a partial order in $\Delta_{n-1}$. This readily explains the existence of incomparable states. 

\subsection{Local convertibility and majorization}
Nielsen's criterion is a result about conclusive LOCC-conversions, i.e.\ about the possibility to find an operation in the class LOCC such that, given pure bipartite states $\ket{\psi},\ket{\varphi}$, the probability of converting $\ket{\psi}$ into $\ket{\varphi}$ is $1$. 
Suppose that a perfect local conversion of  $\ket{\psi}$ into $\ket{\varphi}$ cannot be achieved; is there any procedure realizing the conversion with positive probability of success? And if this is the case, is it possible to determine the maximum probability of success in the conversion? The precise answer to these questions was found by Vidal~\cite{vidal1999} and is expressed in the following theorem.

\begin{thm}[Vidal's criterion~\cite{vidal1999}]
\label{ref:thm:Vidal}
Let  $\ket{\psi},\ket{\varphi}\in\Hi$. Then,
	\begin{equation}\label{eq:VidalProb}
		\Pi\left(\psi \rightarrow \varphi \right)=\Pi(\lambda(\psi_A),\lambda(\varphi_A))
	\end{equation}
			where
	\begin{equation}\label{eq:PIdef}		
		\Pi(x,y)=\min_{1\leqslant k \leqslant n} \frac{\sum_{j=k}^n x_j^{\downarrow}}{\sum_{j=k}^n y_j^{\downarrow}}.
	\end{equation}
\end{thm}

Note that Vidal's theorem is an extension of Nielsen's result. Indeed,
\be
x\prec y\quad\iff\quad \Pi(x,y)=1.
\ee
We can also write (as done originally by Vidal) 
\be
\Pi(x,y)=\min_{1\leqslant k \leqslant n}\frac{E_k(x)}{E_k(y)},
\ee
where the functions $E_k(x)=\sum_{j=k}^n x_j^{\downarrow}$ are entanglement monotones (Schur-concave functions).

We now give a geometric characterization of the function $\Pi(\cdot,\cdot)$.

\begin{prop}
Let $e=(1,0,\dots,0)\in\Delta_{n-1}$. Then,
\begin{equation}
	\Pi(x,y)=\max\left\{0\leqslant p \leqslant 1 :   x^{\downarrow}\prec py^{\downarrow}+(1-p)e\right\}.
	\label{eq:maxp}
\end{equation}
\end{prop}
\begin{proof}
To prove~\eqref{eq:maxp}, {first note that}:
\begin{align}
	\Pi(x,y)&=\max\left\{0\leqslant p \leqslant 1 \colon \sum_{j=k}^npy_k^\downarrow \leqslant  \sum_{{j=k}}^n x_k^\downarrow \quad \text{for $1\leq k\leq n$} \right\}, \label{eq:vid_ineq_1}
\end{align}
and then  observe that
\be
\sum_{j=k}^npy_k^\downarrow \leqslant  \sum_{{j=k}}^n x_k^\downarrow\quad \text{for $1\leq k\leq n$} \iff x^{{\downarrow}} \prec py^{{\downarrow}}+(1-p)e.
\ee
Indeed, the term $(1-p)e$ does not affect the inequalities $\sum_{j=k}^npy_k^\downarrow \leqslant  \sum_{{j=k}}^n x_k^\downarrow$  for $k=2,\dots,n$, while saturates the last inequality for $k=1$, ensuring the normalization condition. 
\end{proof}
This remark gives a geometric picture of the maximal probability of success:  $\ket{\psi}$ cannot be conclusively converted into $\ket{\varphi}$ (i.e., $\Pi\left(\psi \rightarrow \varphi \right)<1$), if and only if $\lambda(\psi_A)$ is not majorized by $\lambda(\varphi_A)$ or, in other words, if $\lambda(\psi_A)$ is not contained in the  polytope given by the convex hull of the permutations of $\lambda(\varphi_A)$. 
The analogue polytope associated to a convex combination  of  $\lambda(\varphi_A)^{\downarrow}$ with the largest element $e$ is always larger than the polytope associated to $\lambda(\varphi_A)^{\downarrow}$. We can therefore consider an enlarged polytope associated to $p\lambda(\varphi_A)^{\downarrow}+(1-p)e$ that contains $\lambda(\psi_A)$ (such a polytope always exists).
 $\Pi(\lambda(\psi_A),\lambda(\varphi_A))$ represents the maximum weight we can assign to $\lambda(\varphi_A)$ in the convex combination such that this is possible. In summary, we have
 \be
\Pi\left(\psi \rightarrow \varphi \right)=\max\left\{0\leqslant p \leqslant 1 :    \lambda(\psi_A)^{\downarrow}\prec p\lambda(\varphi_A)^{\downarrow}+(1-p)e\right\}.
 \ee

\section{Random pure states, Laguerre unitary ensemble and tridiagonal models} 
\label{sec:numericalmethod}
In this section we recall the essential theory of random pure states and we explain the numerical method that made possible the study reported in the present paper. 
For an extensive review on the topic of probability measure on the set of quantum states, see~\cite{zyczkowski2011}. 
The following theorem, essentially due to Lloyd and Pagels~\cite{lloyd1988complexity}, is one of the central results in the theory of random quantum states. For a detailed proof, see~\cite{zyczkowski2001induced}. 
\begin{thm}[Lloyd and Pagels~\cite{lloyd1988complexity}] Let $\ket{\psi}$ be a state distributed according to the unitarily invariant measure on the unit sphere of $\Hi=\Hi_A\otimes\Hi_B$, with $\Hi_A=\C^n$, $\Hi_B= \C^m$, and  $n\leq m$.  
The local spectrum $\lambda(\psi_A)$ is a random $n$-tuple with probability density
\begin{equation}\label{eq:jointEigDistHaar}
g_{n,m}(\lambda)=c_{n,m}\prod_{1\leqslant i<j \leqslant n} (\lambda_i-\lambda_j)^2 \prod_{1\leqslant k \leqslant n} \lambda_k^{m-n}\, 1_{\Delta_{n-1}}(\lambda)
\end{equation}
where $1_{\Delta_{n-1}}$ is the indicator function of the simplex~\eqref{eq:WeylCh}, and 
\begin{equation}
	c_{n,m}=\frac{\Gamma(nm)}{\prod_{j=1}^n \Gamma(j+1)\Gamma(m-n+j)}.
\end{equation}
\end{thm}

A method to sample $n$-tuples from the distribution~\eqref{eq:jointEigDistHaar} comes from results in random matrix theory~\cite{mehta2004random,Vivo}. 
Let $A$ be a $n\times m$ matrix with independent standard complex Gaussian entries. The $n\times n$ random matrix  $AA^\dagger$
is called \emph{complex Wishart matrix} of size $n$ and parameter $m$. 
The joint  distribution of the eigenvalues of    $AA^\dagger$  is~\cite{edel2002mat,Vivo}
\begin{equation}\label{eq:jointEigWishart}
G_{n,m}(\lambda)=C_{n,m}  \prod_{1\leqslant i< j \leqslant n} (\lambda_i-\lambda_j)^2 \prod_{k=1}^n \lambda_{k}^{m-n} e^{-\lambda_k/2} \,1_{\R_+^n}(\lambda)
\end{equation}
where
\begin{equation}
	C_{n,m}=\frac{1}{2^{mn}}\frac{1}{\prod_{j=1}^n \Gamma(j+1) \Gamma (m-n+j)}.
\end{equation}
The probability density~(\ref{eq:jointEigWishart}) is known in random matrix theory as Laguerre unitary ensemble. It can be shown that (see~\ref{appA}) the eigenvalues  of the normalised matrix $AA^\dagger /\Tr(AA^\dagger)$ are distributed according to~\eqref{eq:jointEigDistHaar}. This simple connection between Wishart matrices and random density matrices has been exploited extensively in the previous literature~\cite{Chen10,Cunden13,DePasquale10,DePasquale12,Facchi08,Facchi13,Facchi19,Kubotani08,Kumar11,Kumar17,majumdar2008exact,Nadal10,Osipov20,page1993average,Slater18,Slater19,Vivo10,Vivo11,Vivo16,Znidaric07,zyczkowski2011}, and provides a direct route to sample the local spectrum $\lambda(\psi_A)$ of a  uniformly distributed  pure state $\ket{\psi}\in\C^n\otimes\C^m$. The method proceeds as follows:
\begin{enumerate}
	\item Sample a $n\times m$ random matrix $A$ with independent complex Gaussian entries;
	\item Compute  the eigenvalues of the $n\times n$ random matrix $AA^\dagger /\Tr(AA^\dagger)$.
\end{enumerate}

As already told in advance in the Introduction, there exists a more efficient procedure based on a tridiagonal construction  (see~\cite{edel2002mat}, Theorem 3.4) which is interesting both from a theoretical point of view and for numerical computations~\cite{edelman2013random}.
\begin{thm}[Dumitriu and Edelman~\cite{edel2002mat}]
\label{thm:dum_edel}
	Let $A$ be the $n\times n$ random bidiagonal matrix
	\begin{equation}\label{eq:Am}
	A=
		\begin{pmatrix}
			\chi_{2m} & & & \\
			\ \chi_{2(n-1)} & \ \chi_{2m-2} & & \\
			& \ddots & \ddots& \\
			& & \chi_2 & \chi_{2m-2(n-1)}
		\end{pmatrix},
	\end{equation}
where {the $\chi_\nu$'s are independent chi-distributed (real) random variables} with $\nu$ degrees of freedom. Then, the eigenvalues of  $AA^T$ are distributed according to the Laguerre unitary ensemble~\eqref{eq:jointEigWishart}.
\end{thm}
The theorem above suggests an alternative procedure to sample $\lambda(\psi_A)$:
\begin{enumerate}
	\item Sample a $n\times n$ bidiagonal real random matrix $A$ as in Theorem~\ref{thm:dum_edel};
	\item Compute the eigenvalues of the tridiagonal matrix $A A^T/\Tr(A A^T)$.
\end{enumerate}
In order to understand to which extent this procedure is  more efficient than the ordinary one, let us briefly analyze the computational time and storage resources needed with each method.
\begin{description}
	\item[{Ordinary method}.] To sample $A$ we need to generate and store $nm=cn^2$ complex numbers (this is the number of complex coordinates of a vector in $\C^n\otimes\C^m$). The matrix multiplication $AA^\dagger$ using a standard algorithm requires $\Or(mn^2)=\Or(c \,n^3)$ complex operations, and the diagonalization of $AA^\dagger/\Tr (AA^\dagger)$ has a computational cost of order $\Or(n^3)$  in the field $\C$. 
\end{description}
\begin{description}
	\item[{Tridiagonal method}.] 
	We just need to generate a $n\times n$
bidiagonal matrix $A$  with independent (but not identically distributed) entries as in Theorem~\ref{thm:dum_edel}. Note that $A$ is sparse, and only $\Or(n)$ real numbers need to be stored. Moreover, the diagonalization of the real tridiagonal matrix $AA^\dagger/\Tr (AA^T)$ requires $\Or(n^2)$ operations in the real field instead of $\Or(n^3)$ complex operations. Crucially, the value $m=cn$ enters only as a parameter of the distributions of the entries. Therefore, the computational time does not scale with $c$.
\end{description}
The numerical advantage of the tridiagonal construction is twofold. First, the storage and computational costs scale as $\Or(n)$ and $\Or(n^2)$ where $n$ is the dimension of party $A$. Second, the dimension $m$ of party $B$  is a parameter of the matrix model; this provides a very effective route to sample local spectra of bipartite states from  highly unbalanced bipartitions ($m$ much larger than $n$). 

Using this method we sampled, for several values of $n$ and $m$, $M=5\cdot 10^5$ pairs of spectra $\lambda(\psi_A)$ and $\lambda(\varphi_A)$ distributed according to~\eqref{eq:jointEigDistHaar}. These pairs correspond to the local spectra of states $\ket{\psi}$, $\ket{\varphi}$ from bipartite Hilbert spaces $\Hi$ of dimension ranging from $4$ (two qubits) up to $2^{20}\cdot10^{6}\simeq10^{12}$.  Using this sample we computed numerically the distribution (and other statistics) of $\Pi\left(\psi \rightarrow \varphi \right)=\Pi(\lambda(\psi_A),\lambda(\varphi_A))$, and draw the conclusions listed in the Introduction.

It is clear that the results presented in this paper would not have been possible without the tridiagonal construction (the ordinary method would have required to sample and diagonalise full matrices of size $\simeq10^{12}$). 
We are not aware of previous works on random pure states that made a systematic use of the tridiagonal method. It is our hope that the present manuscript would popularize the tridiagonal trick as a random matrix technique in quantum information theory.

\section{Volume of LOCC-convertible states, Nielsen's conjecture and persistence exponents} \label{sec:incomparability}

In~\cite{nielsen1999} Nielsen conjectured that the probability that two states chosen at random on $\Hi=\C^n\otimes\C^m$ are locally convertible goes to zero as $n,m\to\infty$. Equivalently, the relative volume of LOCC-convertible states of $\Hi$ vanishes in the limit of large dimension. Nielsen's criterion translates the problem of LOCC-convertibility in terms of the majorization relation between their local spectra. Therefore, the conjecture can be stated as
\begin{conj}[Nielsen~\cite{nielsen1999}] Let $x$ and $y$ be independently distributed according to the probability density~\eqref{eq:jointEigDistHaar}. Then
$P(x\prec y)\to0$ as $n,m\to\infty$.
\end{conj}
\par
To justify this conjecture, consider the following argument. For $x, y\in\Delta_{n-1}$, define $\delta=y^\downarrow-x^\downarrow\in\R^n$. Then, the event $x\prec y$ coincides with the event that the process $S_k=\sum_{j=1}^k\delta_j$, starting at $S_0=0$, stays positive
\begin{align}
	P\left(x\prec y\right)=P\left(\sum_{j=1}^k \delta_j\geq0, \,\, \text{for $1\leq k\leq n$}\right) 
	=P\left(\text{$({S}_k)$ stays positive}\right).
\end{align}
Therefore, the probability that $x\prec y$ can be recast as the persistence probability (the probability that a stochastic process does not change sign) of the process  
\begin{align}
S_k&=S_{k-1}+\delta_k,\quad S_0=0, \label{eq:rec}\\
\delta_k&=y_k^\downarrow-x_k^\downarrow\label{eq:steps}
\end{align}
For ``generic'' random vectors $x,y\in\R^n$, it is plausible that $(S_k)$ is unlikely to stay positive at all times if the number of steps $n$ is large.
This statement would be true if $S_k$ were a simple random walk (independent and identically distributed steps $\delta_k$). However, as already remarked by Nielsen~\cite{nielsen1999}, $(S_k)$ is \emph{not} a simple random walk:
\begin{enumerate}
	\item The steps $\delta_k$ are not independent and identically distributed (the components of $x^{\downarrow}_k$ and $y^{\downarrow}_k$ are neither independent nor identically distributed);
	\item The normalization of $x$ and $y$ implies that $(S_k)$ is a `random bridge' starting at $S_0=0$ and ending at $S_n=0$.
\end{enumerate}

It is interesting to compare the process defined by~\eqref{eq:rec}--\eqref{eq:steps} with the classical examples of random walks of the form
\begin{align}
S_k&=S_{k-1}+\eta_k. \label{eq:rec2}
\end{align}
In a general investigation of persistence probability (in modern language) of sums of random variables, Sparre Andersen singled out two exactly solvable cases. Denote by $N_n$ the number of sums $S_1,\dots, S_n$, which are nonnegative.  

\begin{thm}[Sparre Andersen~\cite{Sparre-Andersen}] 
\begin{itemize}
\item[(i)] (Random walks with independent steps) If the steps $\eta_1,\dots,\eta_n$ are independent and identically distributed with continuous distribution symmetric around $0$, then,
\be
P(N_n=k)=(-1)^n\binom{-\frac{1}{2}}{k}\binom{-\frac{1}{2}}{n-k};
\ee
In particular, the persistence probability is
\begin{align}
P\left(\text{$({S}_k)$ stays positive}\right)=P(N_n=n)&=(-1)^n\binom{-\frac{1}{2}}{n}\nonumber\\
&\simeq \frac{1}{\sqrt{\pi n}},\quad \text{as $n\to\infty$}.
\end{align} 
Moreover, the time spent above zero is asymptotically described by the arcsine distribution:
\be
P\left(\frac{N_n}{n}\leq t\right) \rightarrow \frac{2}{\pi}\arcsin\sqrt{t},\text{ as } n\to\infty.
\ee
\item[(ii)] (Random bridges with exchangeable steps).  	If the steps $\eta_1,\dots,\eta_n$ are exchangeable random variables with continuous distribution and ${S}_n=0$ almost surely, then $N_n$ is uniformly distributed over $\{1,\dots,n\}$:
\be
P(N_n=k)=\frac{1}{n};
\ee
 In particular, the persistence probability is
 \be
P\left(\text{$({S}_k)$ stays positive}\right)=P(N_n=n)=\frac{1}{n}.
\ee
\end{itemize}
\end{thm}
In various discrete time processes, the persistence probability turns out to decay algebraically at large times, $P\left(\text{$({S}_k)$ stays positive}\right)\simeq b/n^{\theta}$~\cite{eisenberg1992variation}. The persistence exponent $\theta$ carries information about the full history of the process. For this reason, $\theta$ is usually a nontrivial exponent whose calculation is typically very challenging for non-Markovian processes. For an extensive review, see~\cite{Bray2013}. 
The two classes covered by the Sparre Andersen theorem, are examples of exactly solvable models where the persistence exponent and the prefactor can be computed explicitly ($\theta=1/2$, $b=1/\sqrt{\pi}$ for random walks, and $\theta=1$, $b=1$ for random bridges).

It is interesting to note that in our original problem, if we drop the ordering on the local spectra, then the persistence probability problem changes dramatically and it becomes exactly solvable. Specifically, if we replace~\eqref{eq:steps} with  $\widetilde{\delta}_k=y_k-x_k$, then the process $(S_k)$ becomes a random bridge with exchangeable steps (the common distribution $g_{n,m}$ of $x$ and $y$ is invariant under permutation) and the Sparre Andersen theorem applies. Hence, in such a simplified problem we can show that the persistence probability goes to zero for large $n$ 
(and how fast it goes to zero). Needless to say, majorization is expressed through the ordered vectors $x^{\downarrow}$ and $y^{\downarrow}$, and thus we are interested in the process $(S_k)$ with steps $\delta_k$ rather than the exchangeable steps $\widetilde{\delta}_k$.

To get some insight, we studied numerically for the process~\eqref{eq:rec}--\eqref{eq:steps} the number $N_n$ of sums $S_1,\dots, S_n$, which are nonnegative, i.e.\ the total time spent by the process $(S_k)$ above $0$.  The distribution of $N_n$ is summarised  in \figurename{ \ref {fig:inconvertibility}}, and compared with the uniform distribution (of classical random bridges) and the arcsine distribution (of classical random walks). We see that in the case $c=1$ (balanced partitions $m=n$), the distribution of $N_n$ is close to uniform, but  $(S_k)$ is less likely to cross $0$  compared to a random bridge with exchangeable steps. Indeed, we found numerically that the persistence probability decays with exponent $\theta\simeq 4/5$ (smaller then the exponent $\theta=1$ of classical random bridges).  As $c$ increases, the distribution of $N_n$ looks more and more similar to the arcsine distribution. Again however, $(S_k)$ is less likely to cross $0$  compared to a random walk with independent steps; indeed, the persistence exponent for $c>1$ is $\theta\simeq 2/5$, again smaller that the value $\theta=1/2$ of classical random walks.

In summary, we conclude that the ordering of $x^{\downarrow}$ and $y^{\downarrow}$ defining the steps~\eqref{eq:steps} of the process $(S_k)$, makes the persistence exponent $\theta$ nontrivial (and hard to compute). Nevertheless, there is a sharp distinction in the persistence probability corresponding to balanced ($c=1$) and unbalanced ($c>1$) partitions. This change of behaviour at $c=1$ is consistent with the aforementioned symmetry of the problem under exchange of $m$ and $n$, i.e.\ the symmetry $c\leftrightarrow c^{-1}$.
In the following section, we will make the point that the existence of two distinct persistence exponents is related to the completely different asymptotic behaviour of the smallest eigenvalue $\lambda(\psi_A)_n^{\downarrow}$ when $c$ is strictly larger, or equal to $1$.

\begin{figure}[t!]
	\centering
	\begin{subfigure}[b]{0.49\textwidth}
		\includegraphics[width=\textwidth]{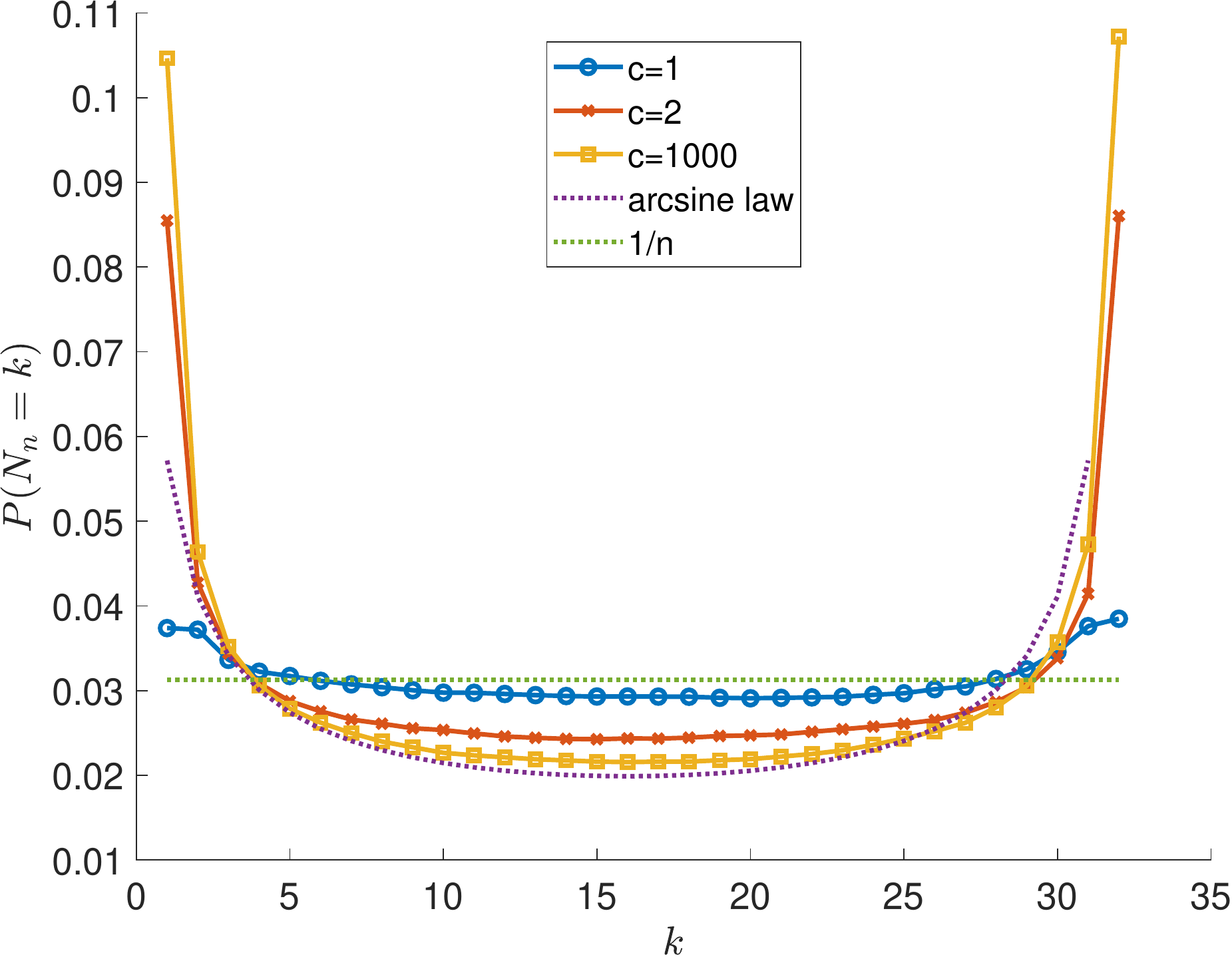}
		\label{fig:ClausesPdf}
	\end{subfigure}%
	~
	\begin{subfigure}[b]{0.49\textwidth}
		\includegraphics[width=\textwidth]{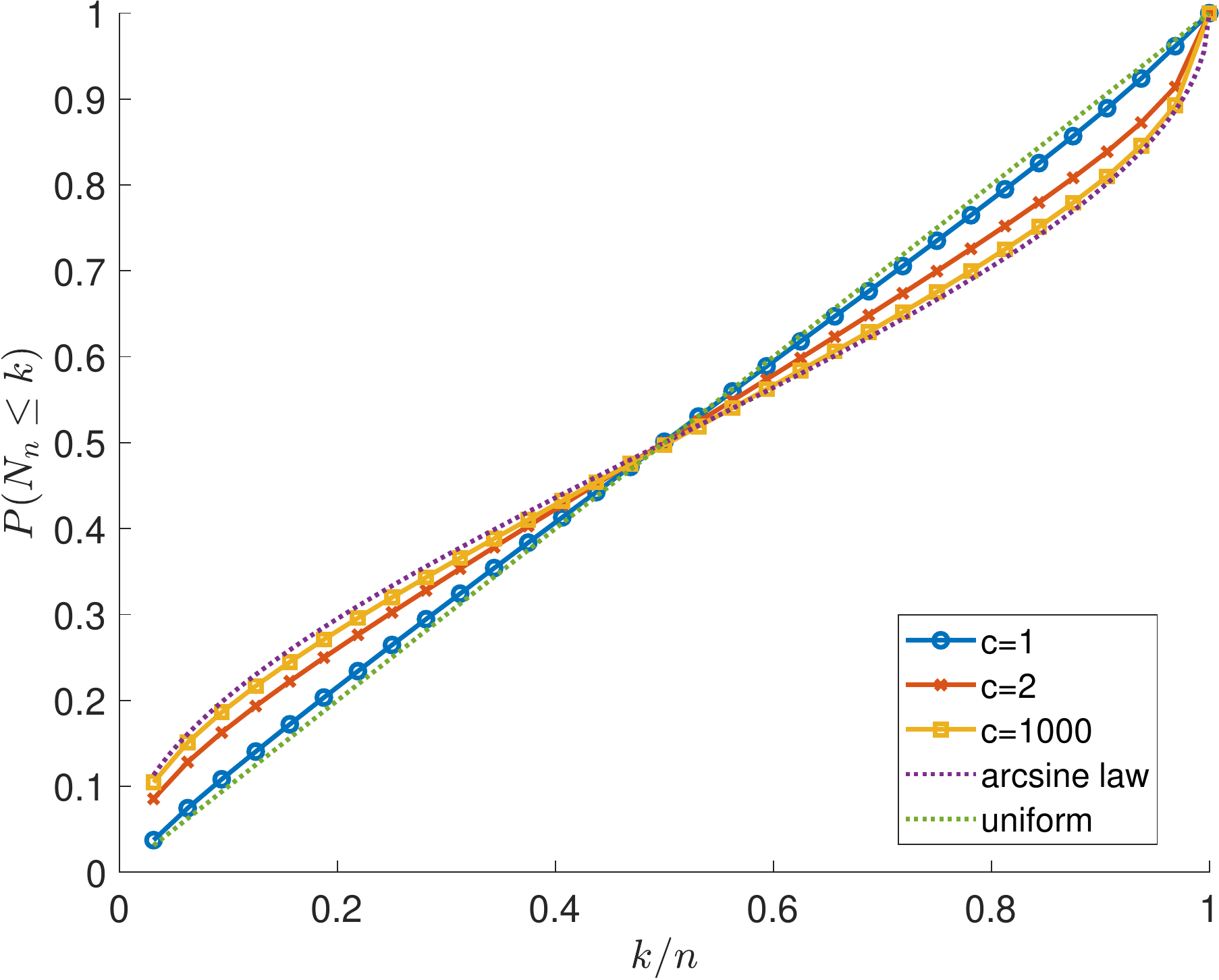}
		\label{fig:ClausesCum}
	\end{subfigure}
	\caption{The probability distribution of $N_n$, the number of times the process $(S_k)$ defined by~\eqref{eq:rec}--\eqref{eq:steps} is above $0$. For comparison, we have inserted the uniform distribution over $l=1,\dots,n$. Here $n=32$.}
	\label{fig:inconvertibility}
\end{figure}

\section{Distribution of the maximal success probability}
\label{sec:results}

In the previous section we considered the probability  
$$P(\ket{\psi}\rightarrow\ket{\varphi})=P\left(\Pi\left(\psi \rightarrow \varphi \right)=1\right)$$ for pairs of random states independent and uniformly distributed on the unit sphere of $\Hi=\C^n\otimes\C^m$. By Nielsen's criterion (Theorem~\ref{thm:Nielsen}), this is the probability $P(x\prec y)$ that two independent probability vectors distributed according to $g_{n,m}$ satisfy the majorization relation; recall that $P(x\prec y)=P(\Pi(x,y)=1)$. In this section we discuss the full distribution $F(p)=P(\Pi(x,y)\leq p)$ of $\Pi(x,y)$. By Vidal's criterion (Theorem~\ref{ref:thm:Vidal}), $F(p)$ is the probability distribution of $\Pi\left(\psi \rightarrow \varphi \right)$, where $\ket{\psi}$ and $\ket{\varphi}$ are pure states picked at random. 

Having established that
\begin{equation}\label{eq:probmaxP=1}
	P\left(\Pi\left(\psi \rightarrow \varphi \right)=1\right)=	P\left(\Pi(x,y)=1\right)=F(1)-F(1^{-})\rightarrow 0 \qquad \text{as }n\rightarrow \infty,
\end{equation}
we studied numerically the asymptotic behaviour of the full distribution of $\Pi(x,y)$. A first natural step in the analysis is the calculation of expectation value 

\begin{equation}
	\mathbb{E}\bigl[\Pi\left(\psi \rightarrow \varphi \right)\bigr] =\mathbb{E}\bigl[ \Pi(\lambda(\psi_A),\lambda(\varphi_A)) \bigr].
	\label{eq:avg_Vid}
\end{equation}
Eq.~\eqref{eq:avg_Vid} is the average maximal probability of successful LOCC conversion between two states picked at random. The behaviour of this expectation value as a function of $n$ and $m$ is summarised in \figurename{ \ref  {fig:Vidal}}.

\begin{figure}[t!]
	\centering
	\includegraphics[width=.6\textwidth]{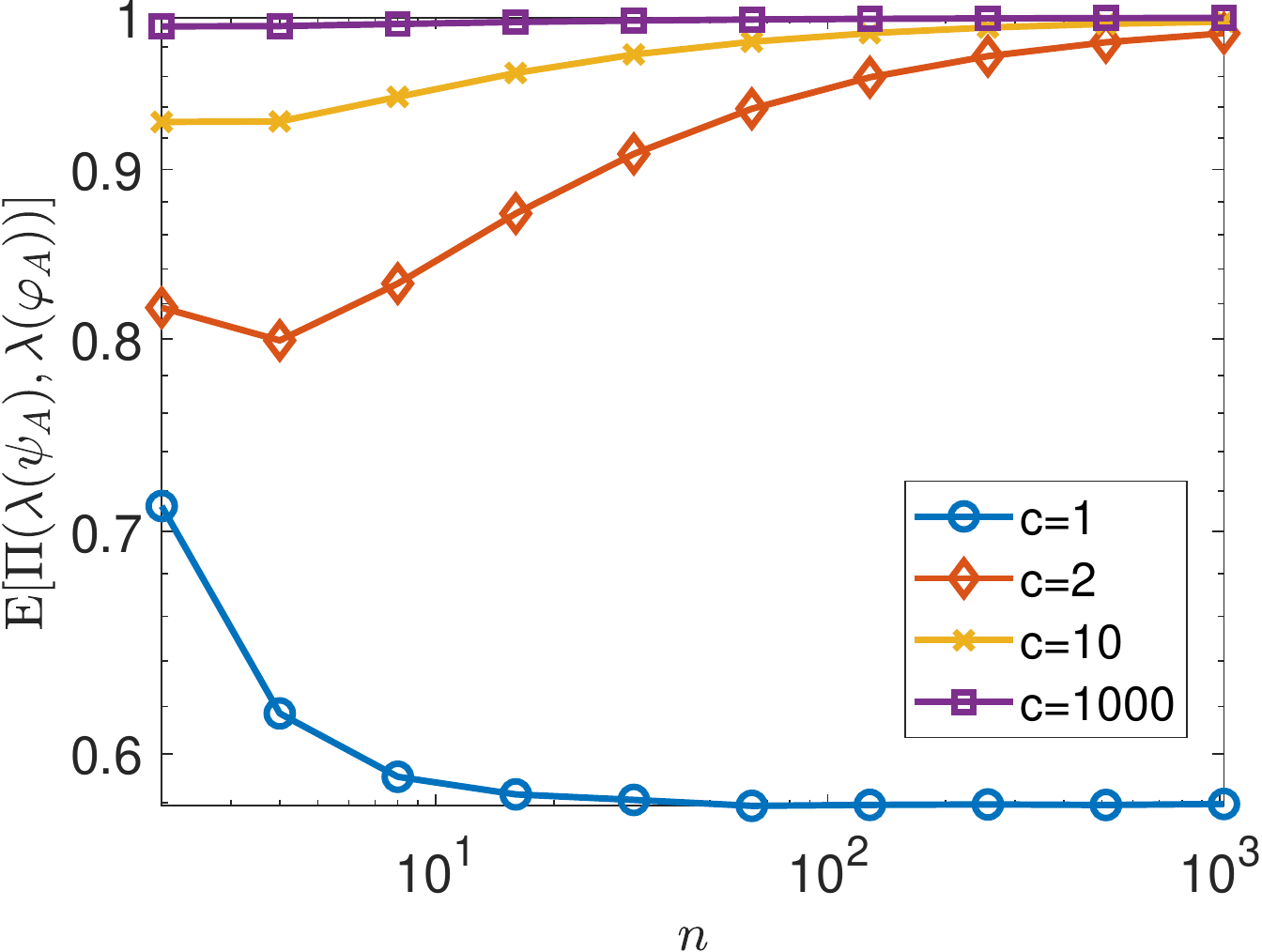}
	\caption{Average value $\mathbb{E}[\Pi\left(\psi \rightarrow \varphi \right)]$ as a function of $n=\dim\Hi_A$. }	
	\label{fig:Vidal}
\end{figure}
We note an interesting and surprising asymptotic behaviour for unbalanced bipartitions ($c>1$). As $n$ and $m$ increase with $c>1$ fixed,  the average maximal probability of conversion eventually approaches the limit value $1$. In formulas:
\begin{equation}\label{eq:averagemaxP}
	\mathbb{E}\bigl[\Pi\left(\psi \rightarrow \varphi \right) \bigr]=\int_0^1 p f(p)\dd p \rightarrow 1 \qquad \text{as } n\rightarrow \infty
\end{equation} 
This is at first surprising given that the probability of successful conversions goes to zero~\eqref{eq:probmaxP=1}. 
Equations~\eqref{eq:averagemaxP} and~\eqref{eq:probmaxP=1} disclose 
an interesting structure of the density $f(p)$ of the maximal probability of conversion. Since for any finite $n$ there is a nonvanishing probability to have an exact conversion, i.e.\ $P(\ket{\psi}\rightarrow \ket{\varphi})=F(1)-F(1^{-})\neq 0$, there is a singular part of the distribution at $p=1$ that goes to zero as $n$ grows. At the same time, equation~\eqref{eq:averagemaxP} tells that there is a continuous part $f_\mathrm{cont}(p)$ which concentrates around (on the left of) $p=1$ as $n\rightarrow \infty$. In the next subsection we will compute exactly this distribution for $n=2$ and $m$ generic to explicitly show  this structure. The calculation for generic $n$ is more complicated, but a numerical computation of this distribution confirms a similar structure. 

We conclude this section by stressing the physical significance of equations~\eqref{eq:averagemaxP} and~\eqref{eq:probmaxP=1}: the relative volume of pairs of pure states that can be \textit{successfully} converted  one into another using local operations vanishes in the asymptotic limit of large dimension; on the other hand, if the bipartition is unbalanced ($c>1$), then the overwhelming majority of states are LOCC-convertible if we allow an arbitrarily small probability of failure in the local conversion. 

\subsection{Exact formulae for $n=2$}
	For $n=2$,  $x\prec y\iff x_1^\downarrow\leqslant y_1^\downarrow$. Since $x$ and $y$ are independent and identically distributed with continuous density $g_{2,m}$, we have
	\begin{equation}\label{eq:remark}
	P(x \prec y)=\frac{1}{2},\quad \text{for $n=2$}.
	\end{equation}

We now compute the probability density $f(p)=F'(p)$ of $\Pi(x,y)$ for $n=2$ and $m\geq n$. 
For $n=2$,
\begin{equation}
	\Pi(x,y)=\min\left\{1,\frac{x_2^\downarrow}{y_2^\downarrow}\right\},
\end{equation}
the function $\Pi(x,y)$ only depends on the ratio of the smallest eigenvalues $x_2^\downarrow$ and $y_2^\downarrow$, and we have
\begin{align}\notag
f(p)&=\int_{\Delta_1} \dd x \int_{\Delta_1} \dd y\, g_{2,m}(x) g_{2,m}(y)\delta\left(p-\min\left\{1,\frac{x_2^\downarrow}{y_2^\downarrow}\right\}\right)\\
\notag
&=\int_0^{1/2} \dd s \int_0^{1/2} \dd t \, q_{2,m}(s) q_{2,m}(t)\delta\left(p-\min\left\{1,\frac{s}{t}\right\}\right)
\end{align}
where $q_{2,m}(s)=2g_{2,m}(1-s,s)$ is
\begin{equation}\label{eq:q_m}
q_{2,m}(s)= \frac{\Gamma(2m)}{\Gamma(m)\Gamma(m-1)}(s-s^2)^{m-2}(1-2s)^2.
\end{equation}
Then, by splitting the integral we get
\begin{align}\notag
	f(p)&=\int_0^{1/2} \dd t \int_t^{1/2} \dd s  \, q_{2,m}(s) q_{2,m}(t)\delta(p-1)+\\\notag
	&\qquad+ \int_0^{1/2} \dd t \int_{0}^{t} \dd s  \, q_{2,m}(s) q_{2,m}(t)\delta\left(p-\frac{s}{t}\right)\\
	&=P(x\succ y) \delta(p-1) +\int_0^{1/2} t\, q_{2,m}(p t) q_{2,m}(t)\dd t \,  1_{[0,1]}(p).
\end{align}
We see then that $f(p)$ is the sum of a singular part and a continuous density
\begin{equation}
	f(p)= f_{\mathrm{cont}}(p)+\frac{1}{2} \delta(p-1), 
\label{eq:n=2Vid.2}
\end{equation}
where we used~\eqref{eq:remark} and we have defined
\begin{equation}
f_{\mathrm{cont}}(p)=1_{[0,1]}(p)\int_{0}^{1/2} t\; q_{2,m}(t)q_{2,m}(pt) \dd t.
\label{eq:n=2Vid2.2}
\end{equation}
Of course, $f_{\mathrm{cont}}(p)\geq0$ and $\int_{0}^1 f_{\mathrm{cont}}(p)dp=1/2$. In fact, for $0\leq p \leq 1$, $ f_{\mathrm{cont}}(p)$ is a polynomial in $p$ of degree $2m-2$. Explicit formulae for the first values of $m$ are
\begin{align}
f_{\mathrm{cont}}(p)&= \frac{3}{20} \left(p^2-4 p+5\right), &m=2, \label{eq:fctnAnalytical1}\\
f_{\mathrm{cont}}(p)&=-\frac{5}{224}p \left(13 p^3-80 p^2+165 p-120\right), &m=3,  \label{eq:fctnAnalytical2}\\
f_{\mathrm{cont}}(p)&=\frac{105}{36608 }p^2 \left(139 p^4-1148 p^3+3549 p^2-4888 p+2574\right), &m=4.  \label{eq:fctnAnalytical3}
\end{align}
In Figure~\ref {fig:VidalThNum} we show a comparison between  numerical results and the analytical expressions \eqref{eq:fctnAnalytical1}--\eqref{eq:fctnAnalytical3}.
It is not hard to show that for large $m$ the continuous density $f_{\mathrm{cont}}(p)$ concentrates at $p=1$, $\lim_{m\to\infty}f_{\mathrm{cont}}(p)=\frac{1}{2}\delta(1-p)$, so that 
\begin{equation}
\lim_{m\to\infty}F(p)=\theta(p-1),\quad\text{for $n=2$}.
\end{equation}

\begin{figure}[t!]
	\centering
	\includegraphics[width=.33\textwidth]{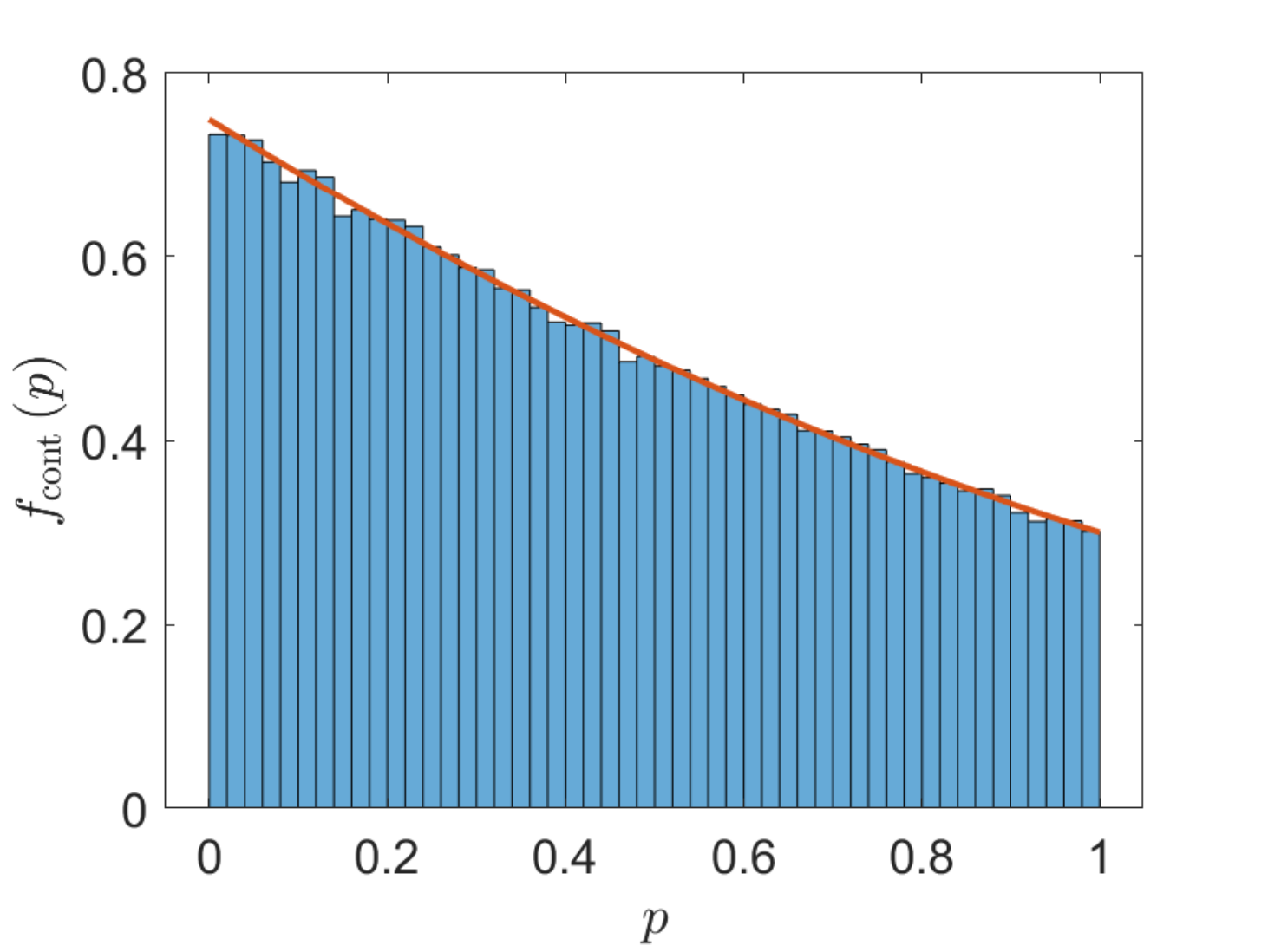}~
	\includegraphics[width=.33\textwidth]{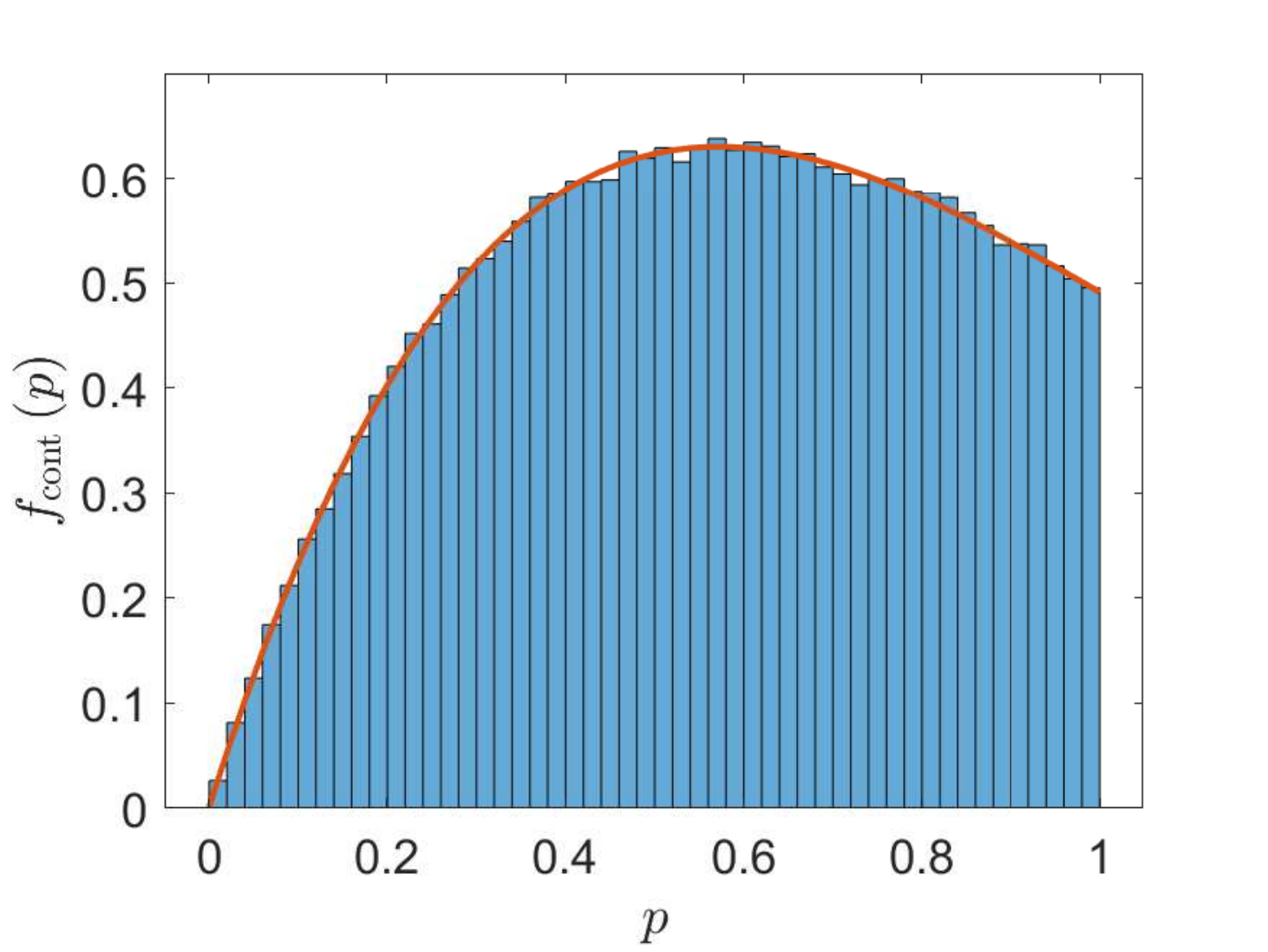}~
	\includegraphics[width=.33\textwidth]{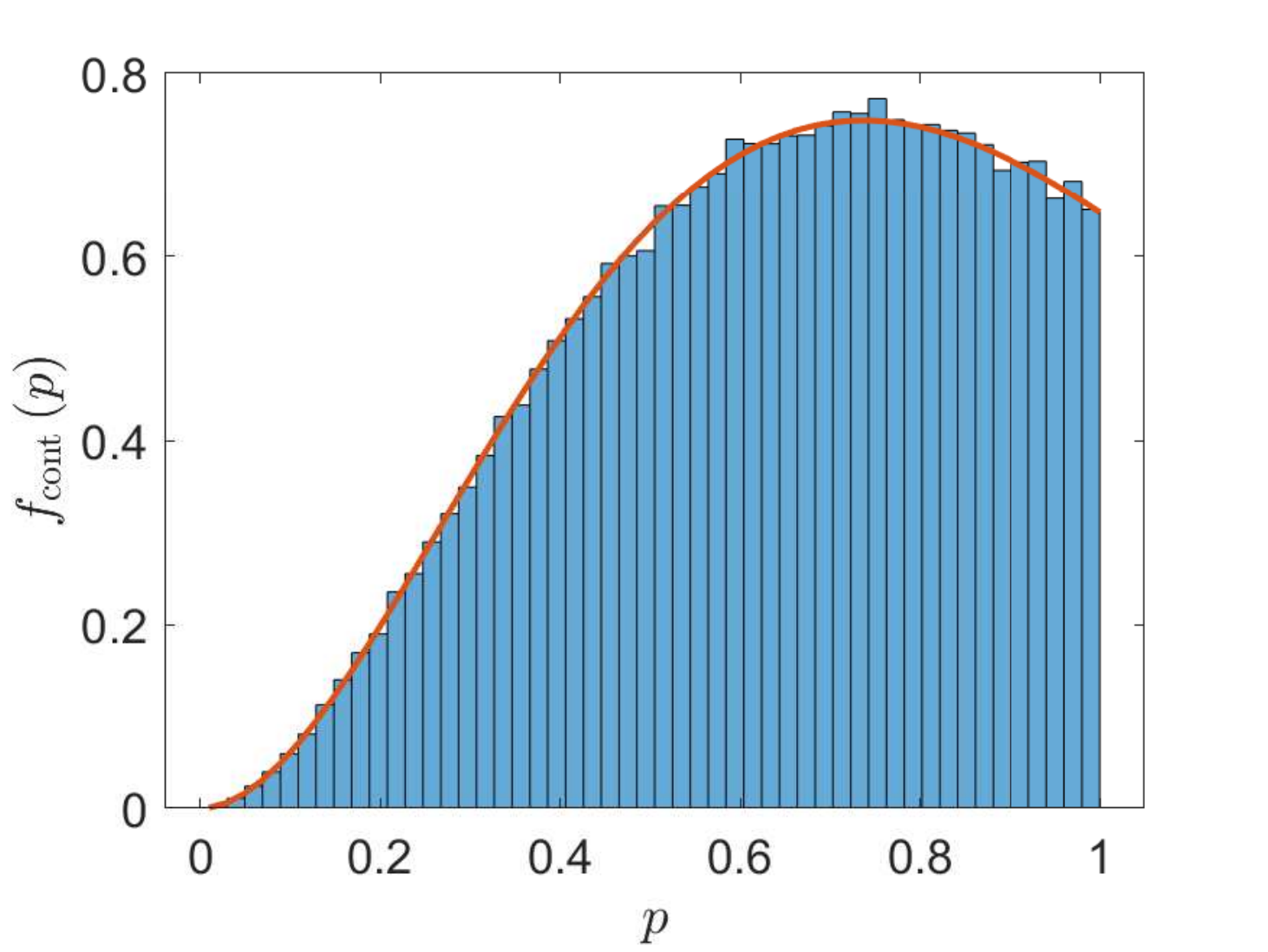}
	\caption{Comparison between numerical results (blue histograms) and the probability density functions (orange curves) $f_{\mathrm{cont}}(p)$ in equations~\eqref{eq:fctnAnalytical1}--\eqref{eq:fctnAnalytical3}.}
	\label{fig:VidalThNum}
\end{figure}

\subsection{Distribution function of $\Pi(x,y)$, concentration of measure and scaling limits}

A numerical computation of the distribution of $\Pi(\lambda(\psi),\lambda(\varphi))$ for $n>2$ has been performed, whose results are shown in Figures~\ref{fig:VidalDist_c1} and~\ref{fig:VidalDist_c2} for $c=1$ and $c>1$ respectively. These numerical results show that the structure of the probability density function~\eqref{eq:n=2Vid.2} is preserved for $n>2$, i.e.\ we have
\begin{equation}
	f(p)=f_\mathrm{cont}(p)+\kappa \delta(p-1)
	\label{eq:dens_dec}
\end{equation}
for every $n$. Clearly the weight of the singular part is associated to the probability of having a successful conversion strategy, i.e.\ $P(\Pi(x,y)=1)=\kappa>0$, and we retrieve in Figures~\ref{fig:VidalDist_c1} and~\ref{fig:VidalDist_c2} that $\kappa\rightarrow0$ as $n\rightarrow\infty$. A remarkable difference can be noted between the balanced  $(c=1)$ and unbalanced $(c>1)$ bipartition. In the large $n$ limit, while for $c=1$ $f_{\mathrm{cont}}$ is supported on the whole interval $[0,1]$, for $c>1$ we see (as announced before) that the support of $f_{\mathrm{cont}}$ concentrates around $p=1$. This difference explains the different behaviour observed in \figurename{ \ref  {fig:Vidal}} between the balanced $(c=1)$ and unbalanced $(c>1)$ cases.
\begin{figure}[t]
	\centering
\begin{subfigure}{0.49\textwidth}
	\centering
	\includegraphics[width=\textwidth]{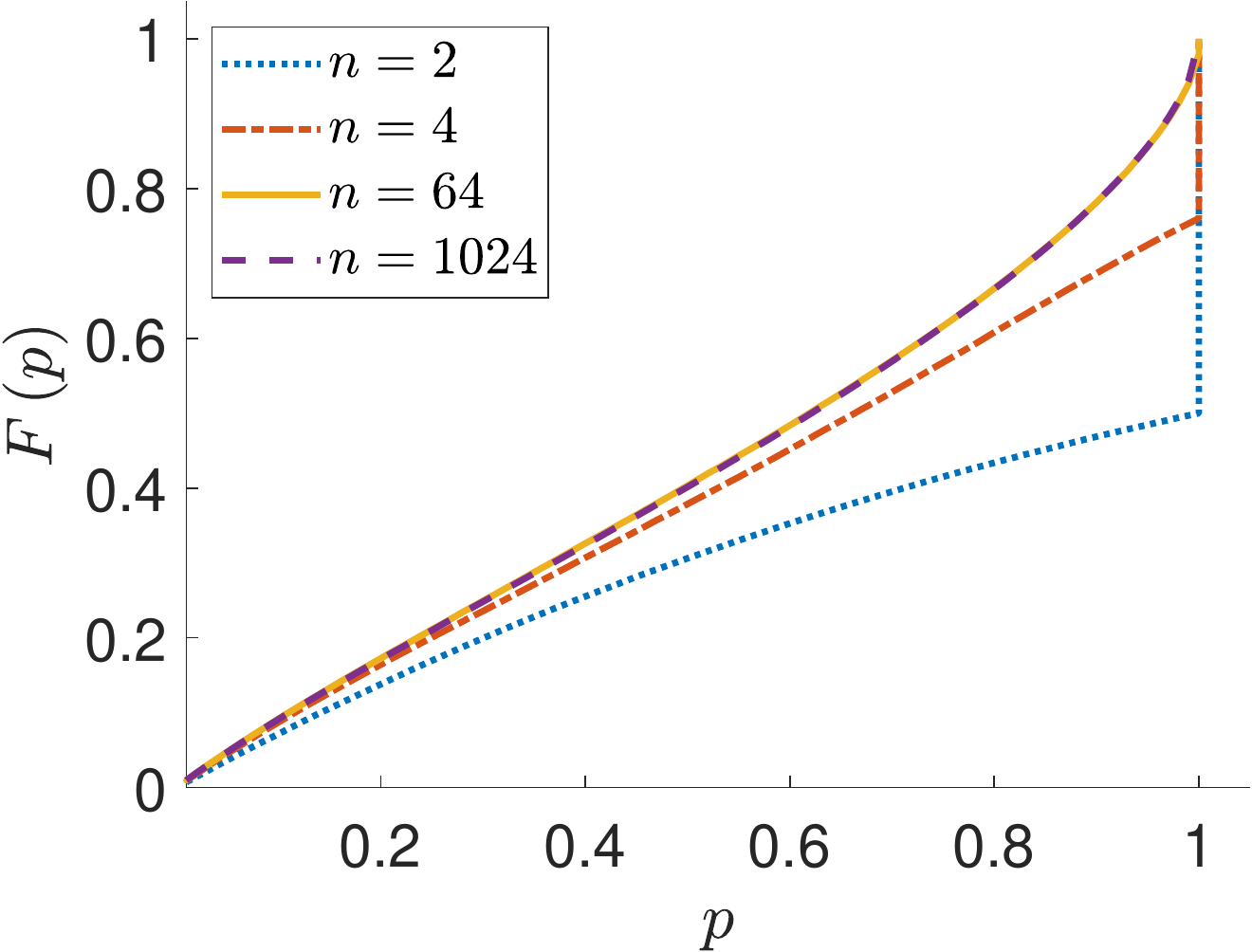}
	\caption{$c=1$}
	\label{fig:VidalDist_c1}
\end{subfigure}%
~
\begin{subfigure}{0.49\textwidth}
	\centering
	\includegraphics[width=\textwidth]{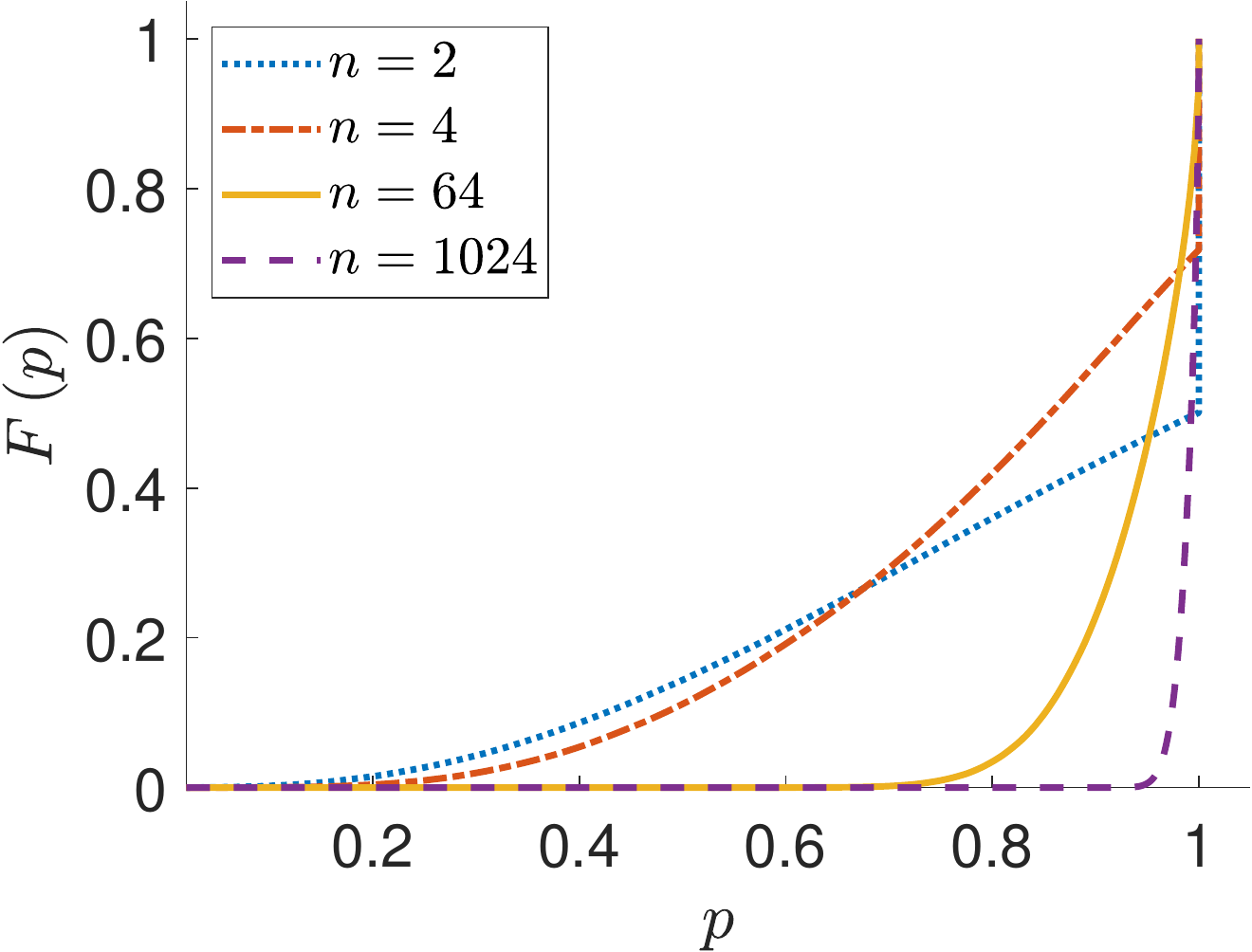}
	\caption{$c=2$}
	\label{fig:VidalDist_c2}
\end{subfigure}%
\caption{Distribution function $F(p)$ of $\Pi\left(\psi \rightarrow \varphi \right)$. The decomposition~\eqref{eq:dens_dec} of the corresponding probability density function $f(p)$  into a singular and a continuous component can be seen from the presence of a jump at $p=1$. In both the balanced and unbalanced cases the height of the jump vanishes as $n$ grows. 
}
\label{fig:VidalDist}		
\end{figure}

We conjecture that the distribution of $\Pi(\lambda(\psi_A),\lambda(\varphi_A))$  is driven by the smallest eigenvalues of $\psi_A$ and $\varphi_A$, i.e.\ $\lambda(\psi_A)_n^{\downarrow}$, $\lambda(\varphi_A)_n^{\downarrow}$. 
Recall that
\begin{align}
\Pi(\lambda(\psi_A),\lambda(\varphi_A))
&=\min\left\{ \frac{\lambda(\psi_A)_n^{\downarrow}}{ \lambda(\varphi_A)_n^{\downarrow}}, \frac{\lambda(\psi_A)_n^{\downarrow}+\lambda(\psi_A)_{n-1}^{\downarrow}}{ \lambda(\varphi_A)_n^{\downarrow}+\lambda(\varphi_A)_{n-1}^{\downarrow}},\dots, \frac{\lambda(\psi_A)_n^{\downarrow}+\cdots+\lambda(\psi_A)_{2}^{\downarrow}}{ \lambda(\varphi_A)_n^{\downarrow}+\cdots+\lambda(\varphi_A)_{2}^{\downarrow}},1\right\}
\label{eq:PIdef.2}
\end{align}
 If each component of  $\lambda(\psi_A)^{\downarrow}$, $\lambda(\varphi_A)^{\downarrow}$ does not vary too much with respect to its mean, then all the sums in the numerators are very close to the sums in the denominators of~\eqref{eq:PIdef.2}, implying that the minimum  of their ratios is unlikely to be much less than one. In order to test this idea, we computed numerically the variance $\operatorname{Var}[\lambda(\psi_A)_k^{\downarrow}]$ and the mean $\mathbb{E}[\lambda(\psi_A)_k^{\downarrow}]$  of the $k$-th largest eigenvalue
 of the density matrix $\psi_A$, and its relative fluctuation 
 \be
 \sqrt{\operatorname{Var}[\lambda(\psi_A)_k^{\downarrow}]}\,/\,\mathbb{E}[\lambda(\psi_A)_k^{\downarrow}].
 \label{eq:rel_fluct}
 \ee 
	In  \figurename{ \ref  {fig:moments}} we plotted the relative fluctuations~\eqref{eq:rel_fluct} for large $n$ and several values of $c\geq1$. 
\begin{figure}
	\qquad \qquad
	\includegraphics[width=.7\textwidth]{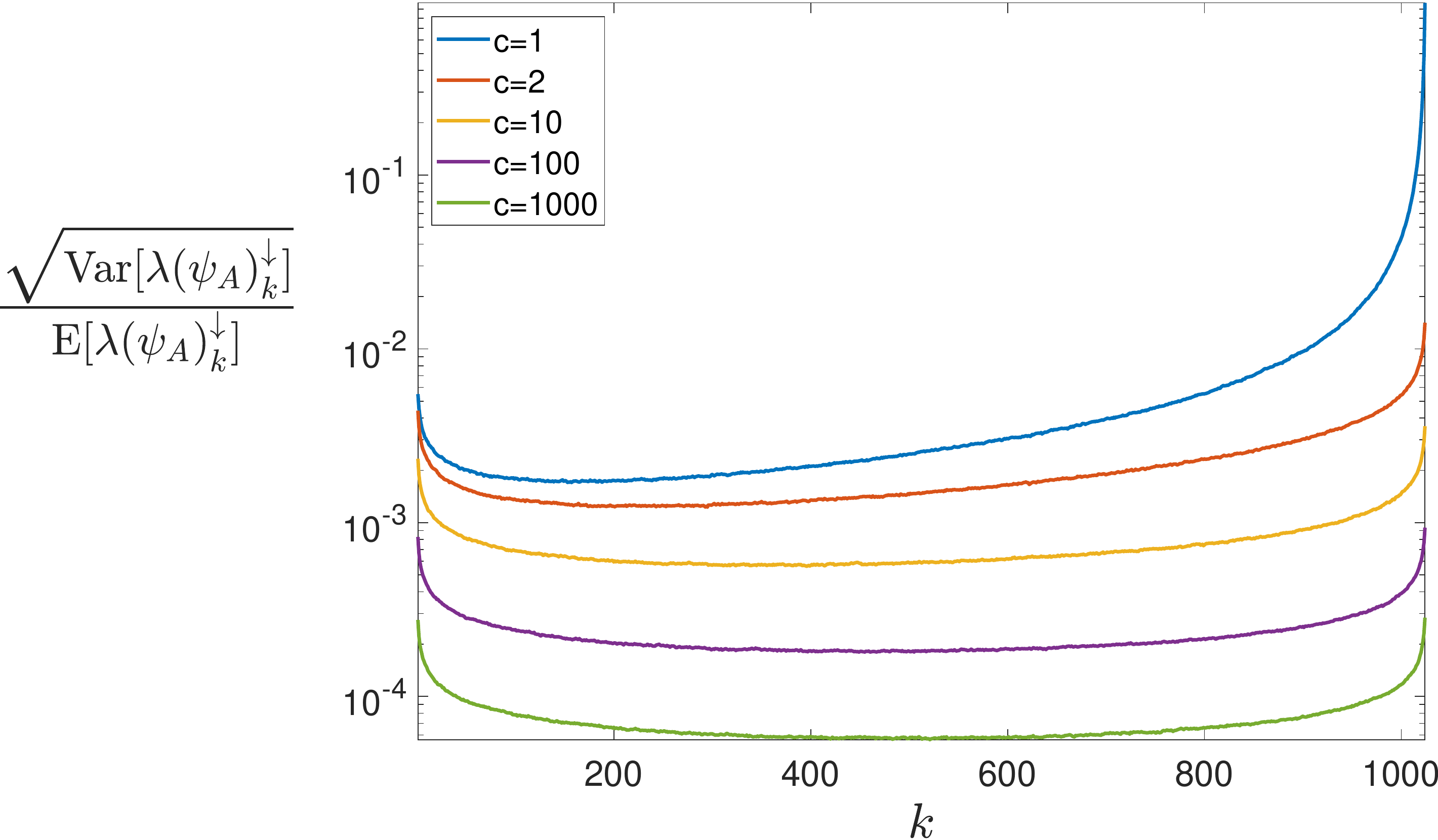}
	\caption{Relative fluctuations $\sqrt{\operatorname{Var}[\lambda(\psi_A)_k^{\downarrow}]}\,/\,\mathbb{E}[\lambda(\psi_A)_k^{\downarrow}]$ of the $k$-th largest eigenvalue for $k=1,\dots,n$. Here $n=1024$. 
	}
	\label{fig:moments}
\end{figure}
The plots show that the main difference between the balanced and the unbalanced cases are in the fluctuations of the smallest eigenvalues, in particular the minimum eigenvalue $\lambda(\psi_A)_n^{\downarrow}$. This is somehow expected. 
Indeed, for large $n$, the rescaled spectrum  $n\lambda(\psi_A)$ concentrates on a bounded interval $[a,b]\subset\R_+$, and has a continuous limit density known as Mar\v{c}enko-Pastur distribution
\be
\rho_c(x)=\frac{c}{2\pi x}\sqrt{(x-a)(b-x)}1_{[a.b]}(x),
\label{eq:MP_distribution}
\ee
where $a=(1-1/\sqrt{c})^2$ and $b=(1+1/\sqrt{c})^2$. When $c>1$, the support of this density is bounded away from zero ($a>0$) and vanishes at the lower edge as $(x-a)^{1/2}$ (soft edge); for $c=1$, the support of the Mar\v{c}enko-Pastur distribution contains the origin and diverges at $a=0$ as $x^{-1/2}$ (hard edge). See \figurename{ \ref  {fig:Marcenko-Pastur}}.

\begin{figure}\centering
	\includegraphics[width=.49\textwidth]{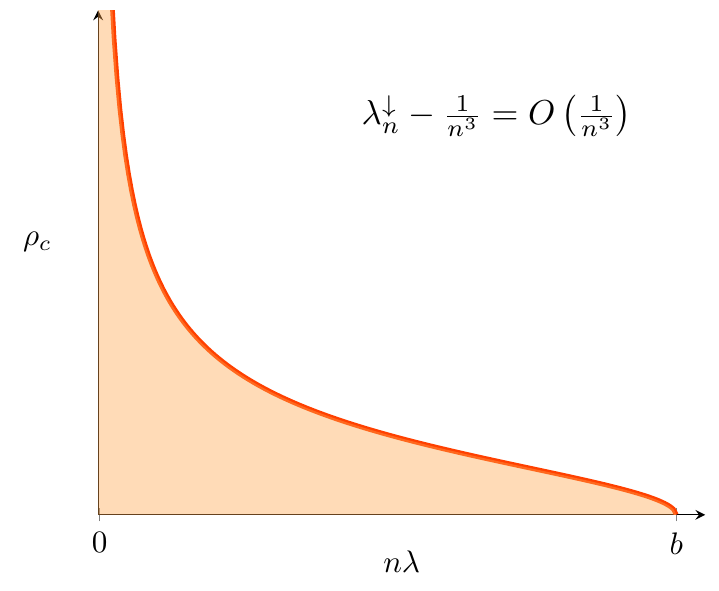}
	\includegraphics[width=.49\textwidth]{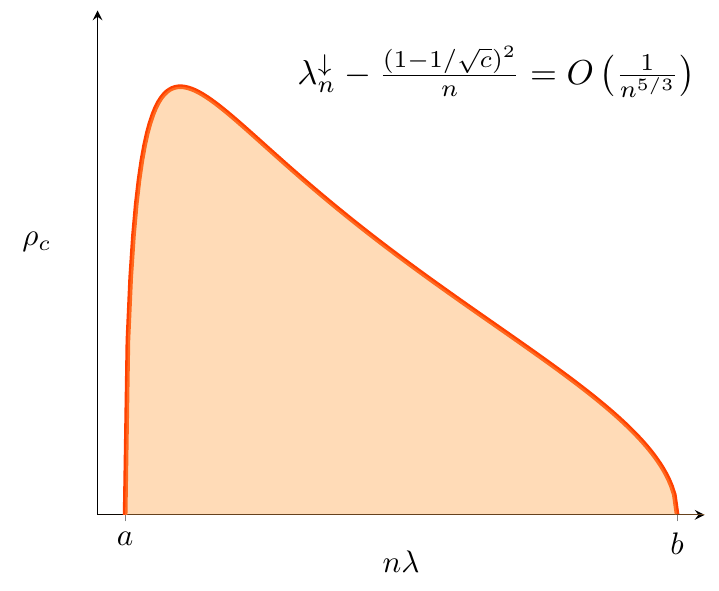}
	\caption{Mar\v{c}enko-Pastur distribution \eqref{eq:MP_distribution} with hard-edge for $c=1$ (left) and soft edge for $c>1$ (right)}
	\label{fig:Marcenko-Pastur}
\end{figure}

For large $n$, the fluctuations of the minimum eigenvalue $\lambda(\psi_A)_n^{\downarrow}$ at the soft or hard edges are
	\begin{equation}
	\frac{\sqrt{\operatorname{Var}[\lambda(\psi_A)_n^{\downarrow}]}}{\mathbb{E}[\lambda(\psi_A)_n^{\downarrow}]}\sim 
	\begin{cases}
	1 &\text{if $c=1$ (hard edge)},\\
	\dfrac{1}{c^{1/6}|1-\sqrt{c}|^{2/3}}n^{-2/3} &\text{if $c>1$ (soft edge)}.
	\end{cases}
	\label{eq:rel_flut}
	\end{equation}
Formula~\eqref{eq:rel_flut} can be obtained by combining previously known results. See~\ref{app:2} for the details.
	
The fluctuations of the smallest eigenvalue are comparable with its mean in the balanced case, suggesting that the ratio
\begin{equation}\label{eq:argminProbVidal}
\frac{\sum_{j=k}^n\lambda_j(\psi_A)^\downarrow}{\sum_{j=k}^n \lambda_j(\varphi_A)^\downarrow},
\end{equation}
takes values significantly smaller than one with a non-negligible probability for $k=n$ (the ratio between the smallest eigenvalues);  the same must happen to $\Pi(\lambda(\psi_A),\lambda(\varphi_A))$, since it is the minimum of~\eqref{eq:argminProbVidal} over  $1\leq k\leq n$. 
On the contrary, for unbalanced bipartitions the relative deviation of the {smallest} eigenvalue  vanishes for large $n$, and as shown in \figurename{ \ref  {fig:moments}}(a) the same happens for all the other eigenvalues. This implies that the ratio~\eqref{eq:argminProbVidal} is typically nearly equal to $1$ for every $1\leq k\leq n$. 

As a final test of our conjectural relation between $\Pi(\lambda(\psi_A),\lambda(\varphi_A))$ and the fluctuations of the minimum eigenvalues $\lambda(\psi_A)_n^{\downarrow}$, $\lambda(\varphi_A)_n^{\downarrow}$, we ask whether the distribution function $F(p)$ has a scaling limit driven by the fluctuations~\eqref{eq:rel_flut}. Specifically, we established that 
\be
\lim_{\substack{n,m\to\infty\\c=m/n \text{ fixed}}}P\bigl(\Pi(\lambda(\psi_A),\lambda(\varphi_A))\leq p\bigr)=\theta(1-p).
\ee
Is it possible to rescale the random variable $\Pi(\lambda(\psi_A),\lambda(\varphi_A))$ in order to obtain a nontrivial scaling limit? In particular, we would like to centre the variable around its limit value $1$ and scale by its typical fluctuations. The previous discussion suggests the rescaling
\begin{multline}
\widetilde{\Pi}(\lambda(\psi_A),\lambda(\varphi_A))=
\frac{\mathbb{E}[\lambda(\psi_A)_n^{\downarrow}]}{\sqrt{\operatorname{Var}[\lambda(\psi_A)_n^{\downarrow}]}}
\bigl(1-\Pi(\lambda(\psi_A),\lambda(\varphi_A))\bigr)\\
\sim 
	\begin{cases}
	\bigl(1-\Pi(\lambda(\psi_A),\lambda(\varphi_A))\bigr) &\text{if $c=1$ (hard edge)}\\
	c^{1/6}|1-\sqrt{c}|^{2/3}n^{2/3}\bigl(1-\Pi(\lambda(\psi_A),\lambda(\varphi_A))\bigr) &\text{if $c>1$ (soft edge)}
	\end{cases}.
	\label{eq:rel_flut_scal}
\end{multline}

It turns out that the rescaled variable $\widetilde{\Pi}$ has a nontrivial limit in distribution. The results of the numerical computation are rather convincing and are summarised in \figurename{ \ref  {fig:VidalRescaledintro}}. This is the content of~\eqref{eq:scaling_unb}--\eqref{eq:scaling_bal} of item (vi) in the Introduction.   A precise description of the limit distributions $H_{\mathrm{bal}}$ and $H_{\mathrm{unb}}$ of $\widetilde{\Pi}$ seems quite challenging.

\section{Conclusions and Outlook}\label{sec:conclusion}
In this paper we have studied the probability of LOCC-convertibility between pairs of random pure states. We have {found strong numerical evidence}
that the proportion of LOCC-convertible states of a bipartite system $AB$ vanishes in the asymptotic limit of large dimensional local Hilbert spaces, thus supporting Nielsen's conjecture. By mapping the problem to the persistence probability of a discrete-time continuous-space random walk, we argued that the probability $P\left(\ket{\psi}\to\ket{\varphi}\right)$ of successful LOCC-conversions decays algebraically $\simeq1/n^{\theta}$, and we gave precise estimates for the persistence exponents $\theta$. It turns out that the value of $\theta$ depends on whether the bipartition is balanced $m\sim n$ ($\theta\simeq 4/5$) or unbalanced $m\sim cn$ with $c>1$ ($\theta\simeq 2/5$). Recall that the problem of LOCC-convertibility is symmetric under exchange of parties $A$ and $B$, and the special value $c=1$ is the fixed point of the symmetry $c\leftrightarrow c^{-1}$. The sharp transition from $c=1$ to $c>1$ can also be  observed in the full distribution $F(p)$ of $\Pi\left(\psi \rightarrow \varphi \right)$. Surprisingly enough, for unbalanced bipartitions ($c>1$), although the probability of successful LOCC-conversion between random states vanishes asymptotically, the overwhelming majority of pairs can be LOCC-converted if any fixed arbitrary small probability of error is allowed. Random matrix theory turns out to be a useful tool to explain this phenomenon and elaborate on it. The main difference between the two regimes is that the density of eigenvalues $\rho_c(x)$ in~\eqref{eq:MP_distribution} accumulates at zero for $c=1$, while it has support bounded away from zero for $c>1$. The fluctuations of the minimum eigenvalues have different scales in the two cases, and we showed that those fluctuations drive the scaling limit of $F(p)$. See Figs.~\ref{fig:Rescaling_bal}--\ref{fig:Rescaling_unbal}. 
\par
Mapping the majorization relation between random vectors to persistence probabilities of random walks suggests that Nielsen's conjecture is an instance of a more general phenomenon. We conjecture that  if $x,y$ are independently distributed from a continuous and permutation invariant distribution on $\Delta_{n-1}$, then $P(x\prec y)\to0$ as $n\to\infty$. One could also extend the problem to discrete measures on the simplex; in fact, while completing this work we learned that an equivalent problem for the dominance order of random integer partitions is stated in Macdonald's book~\cite[Chap. 1, Ex. 18]{Macdonald}. Preliminary investigations indicate that the general conjecture holds (certainly for most `natural' continuous and discrete measures on the simplex), but a more careful analysis is left to future works.
\par
The key numerical ingredient that made possible our analysis is an adaptation of the tridiagonal realisation of the Laguerre unitary ensemble. It is worth emphasising that the general version of the tridiagonal method of Dumitriu and Edelman provides a way to sample $n$-tuples for the Laguerre $\beta$-ensemble for generic $\beta>0$. In particular, one can easily adapt the method explained in Section~\ref{sec:numericalmethod} to investigate the set of random \emph{real} pure states (corresponding to $\beta=1$),  a set that at first seems artificial, but which is in fact quite realistic for many situations (see~\cite{Vivo10,Wooters90,zyczkowski03,zyczkowski2011,zyczkowskibeng}). In that case we expect a similar behaviour as for complex random states, but the persistence exponents might change.
\par
Another possible direction to extend this work is to investigate the convertibility of random states for other quantum resource theories (different from entanglement theory) where the majorization relation provides a necessary and sufficient condition. Perhaps the first theory that one can successfully explore is the `resource theory of coherence' developed in~\cite{Baumgratz14,Du15,Winter16}. 
\par
We conclude with a final remark on the ordered sums of the eigenvalues of random matrices. Natural `observables' in study of spectral properties of random matrices are symmetric under permutations of the eigenvalues (these include moments, correlation functions, averages of characteristic polynomials, etc.). Therefore, the ordered sums $\lambda^{\downarrow}_1+\dots+\lambda^{\downarrow}_k$ might seem rather exotic for a random matrix theorist. Note however that, for the Laguerre unitary ensemble, the ordered sums $\lambda^{\downarrow}_1+\dots+\lambda^{\downarrow}_k$ are the outputs of the Robinson-Schensted-Knuth algorithm applied to an array of i.i.d. exponential waiting times~\cite{Borodin08,Johansson00,O'Connell14}. Therefore, as a  consequence of Greene's theorem~\cite{Greene74}, the ordered sums have a very explicit and direct relation to last passage percolation models. (Related quantities, dubbed `truncated linear statistics', were recently studied in~\cite{Grabsch17}.)  There is no space here to elaborate on the representation theory behind this correspondence; the mere purpose of this remark is to stress the relevance of ordered sums of eigenvalues of random matrices.

\ack
The research of FDC is supported by ERC Advanced Grant 669306.
PF, GF and GG are partially supported by Istituto Nazionale di Fisica Nucleare (INFN) through the project ``QUANTUM".  FDC, PF, GF and GG are partially supported by the Italian National Group of Mathematical Physics (GNFM-INdAM). GF is supported by the Italian Ministry MIUR-PRIN project “Mathematics of active materials: From mechanobiology to smart devices” and by the FFABR research grant.

\appendix

\section{Some background on LOCC}\label{appLOCC}
Despite having a quite simple physical description, the precise mathematical characterization of LOCC is notoriously challenging~\cite{Bennet99,DHR02}, due to the fact that the particular choice of local measurements to perform at some stage of a LOCC protocol depends on all the measurement results previously obtained (and shared between the parties). A further complication arises from the fact that the number of communication rounds is potentially unbounded~\cite{Chitambar11}. A detailed mathematical description of the full LOCC class has been presented in~\cite{Chitambar2014} adopting the quantum instruments formalism. A full treatment of the subject is beyond the scope of this work, and the interested reader is referred to~\cite{Chitambar2014} and references therein. The aim of this appendix is to provide, as shortly as possible, some minimal background suitable for the problem studied in this work for the reader who is not familiar with the subject. 

As a matter of fact, the description greatly simplifies in the case at study here. As proved in~\cite{Lo2001}, any LOCC transformation on a pure bipartite state can be performed in a single-step one-way-communications LOCC protocol, i.e. a single generalized measurement performed by one party followed by the communication of the result to the other party and finally by a local unitary transformation based on the measurement outcome.
Thus, the most general LOCC transformation $T$ on a pure bipartite state $\ket{\psi}$ in $\Hi_A\otimes \Hi_B$ can be written as
\begin{equation}\label{eq:genLOCC1}
	T(\psi)=\sum_{\alpha}(M_\alpha\otimes U_\alpha)\psi ({M_\alpha}^\dagger\otimes {U_\alpha}^\dagger),
\end{equation}
where the operators $M_\alpha$, with  $\sum_\alpha M_\alpha^\dagger M_\alpha=\mathbb{I}$,  describe a generalised measurement on party $A$ with outcomes labelled by $\alpha$, while $U_\alpha$ is the unitary performed on party $B$ when the result of the measurement is $\alpha$. Notice that the operation~\eqref{eq:genLOCC1} is indeed a convex combination of several possible results:
\begin{equation}
	T(\psi)=\sum_{\alpha}p_\alpha \rho_\alpha,
\end{equation}
where 
\begin{equation}
p_\alpha=\Tr(M_\alpha^{\dagger}M_\alpha \psi_A), \qquad \rho_\alpha=(M_\alpha\otimes U_\alpha)\psi ({M_\alpha}^\dagger\otimes {U_\alpha}^\dagger)/p_\alpha,
\end{equation}
with $\rho_\alpha$ being the final state when the protocol is completed, given that the outcome $\alpha$ is obtained during the measurement stage. 
Therefore, the probability of success $p$ of the conversion $\ket{\psi}\xrightarrow{p}\ket{\varphi}$ by a LOCC of the form~\eqref{eq:genLOCC1} is given by
\begin{equation}
p(\psi,\varphi;T) = \sum_{\alpha\colon\rho_\alpha = \varphi} p_\alpha.
\end{equation}
Thus, the maximal conversion probability introduced in Definition 1 can be written using this notation as
\begin{equation}
	\Pi(\psi\rightarrow \varphi)=\max_{T\in\LOCC}p(\psi,\varphi;T)
\end{equation}

As pointed out in~\cite{Chitambar2014}, the set of LOCC  is convex but it is not topologically closed, in the sense that there are some convergent sequences of LOCC protocols (e.g. approaching some probability of success $p$ as more rounds of communications are allowed) which are not LOCC-feasible in the limit~\cite{Chitambar12}. However, this does not pose a problem in the definition of the maximal success probability $\Pi(\psi\rightarrow \varphi)$, since the sets of finite-round LOCC protocols (such as the one-round LOCC protocols we are interested here) are closed~\cite{Chitambar2014}.
\section{Local spectra of random pure states}\label{appA}
We collect here a few facts on the local spectra of random pure states (i.e.\ spectra of random density matrices) and their relation to the Laguerre unitary ensemble. A comprehensive reference on the subject is the paper by Nechita~\cite{Nechita08} (that, among other things, contains a collection of exact formulae on moments and precise asymptotic theorems on the largest eigenvalue of random density matrices). 
\subsection{Eigenvalues of Wishart and fixed-trace Wishart matrices}
The basic connection between local spectra of random pure states and Wishart matrices is as follows.
	Let $(x_1,\dots,x_n)$ be a  vector distributed according to the Laguerre unitary ensemble $G_{n,m}$ in~\eqref{eq:jointEigWishart}. Define
	\begin{equation}
		\lambda_{j}=\frac{x_j}{\sum_{k=1}^n x_k}, \qquad j=1,\dots,n.
		\label{eq:rep_lambda_x}
	\end{equation}
Then $(\lambda_1,\dots,\lambda_n)$ is distributed according to $g_{n,m}$ in~\eqref{eq:jointEigDistHaar}. The proof hinges on the change of variable
\begin{equation}
(x_1,\dots,x_n)\mapsto(\lambda_1,\dots,\lambda_{n-1},s)=(x_1/s,\dots,x_{n-1}/s,s),
\end{equation}
where $s=\sum_{k=1}^n x_k$. 
The Jacobian of this transformation is
\begin{equation}
\left\lvert\frac{\partial (x_1,\dots,x_n)}{\partial (\lambda_1,\dots,\lambda_{n-1},s)}\right\rvert=
\begin{vmatrix}
s & 0 &\dots  & \lambda_1\\
0   & s  &\dots & \lambda_2\\
\vdots &\vdots &\ddots & \vdots\\
-s & -s & \dots & 1-\sum_{k=1}^{n-1}\lambda_k
\end{vmatrix}=
\begin{vmatrix}
s & 0 &\dots  & \lambda_1\\
0   & \ddots  &\dots & \lambda_2\\
\vdots &\vdots &s & \vdots\\
0 & 0 & \dots & 1
\end{vmatrix}=s^{n-1},
\end{equation}
where in the second equality we have added to the last row the first $n-1$ rows, an operation that does not change the determinant. If we denote $\lambda_n=x_n/s=1-\sum_{k=1}^{n-1}\lambda_k$, we see that the probability density function of $(\lambda_1,\dots,\lambda_{n-1},s)$ factorises
\begin{align}\notag
\widetilde{G}_{n,m}(\lambda_1,\dots,\lambda_{n-1},s)
&=C_{n,m} e^{-\frac{s}{2}} s^{n-1}\prod_{1\leqslant i< j \leqslant n} (s\lambda_i-s\lambda_j)^2  \prod_{k=1}^n (s\lambda_k)^{m-n}  \\
&=C_{n,m}  e^{-\frac{s}{2}}s^{nm-1} \left(\prod_{1\leqslant i< j \leqslant n} (\lambda_i-\lambda_j)^2 \prod_{k=1}^n \lambda_k^{m-n}\right)
\end{align}
 Integrating over $s$ we obtain the distribution of the vector $(\lambda_1,\dots,\lambda_{n-1})$. The integral gives the constant
 \be
 C_{n,m} \left(\int_{0}^\infty e^{-\frac{s}{2}}s^{nm-1}\ud s\right) =C_{n,m} 2^{nm} \Gamma(nm)=c_{n,m}.
 \ee
 Therefore, the $n$-tuple $(\lambda_1,\dots,\lambda_{n-1},\lambda_n)$, with  $\lambda_n=1-\sum_{k=1}^{n-1}\lambda_k$ has density 
	\begin{align}
	g_{n,m}(\lambda_1,\dots,\lambda_{n})&=c_{n,m} \prod_{1\leqslant i< j \leqslant n} (\lambda_i-\lambda_j)^2 \prod_{k=1}^n \lambda_k^{m-n}1_{\Delta_{n-1}}(\lambda).
	\end{align}
	with respect to the Lebesgue measure on the simplex.
\subsection{Smallest eigenvalue of fixed-trace Wishart matrices}
\label{app:2}
Majumdar, Bohigas and Lakshminarayan~\cite{majumdar2008exact} computed the distribution and moments of the smallest eigenvalue $\lambda_n^{\downarrow}$ of complex fixed-trace Wishart matrices in terms of elementary functions in the case $m=n$. Let $(\lambda_1,\dots,\lambda_n)$ be distributed according to $g_{n,n}$ in~\eqref{eq:jointEigDistHaar}. Then, $\lambda_{n}^{\downarrow}$ has probability density
\be
f_{\operatorname{min}}(x)=n(n^2-1)(1-nx)^{n^2-2}\theta(1/n-x).
\label{eq:MBL}
\ee
The $k$-th moment is given by
\be
\mathbb{E}\bigl[\left(\lambda_{n}^{\downarrow}\right)^k\bigr]=\frac{\Gamma(n^2)\Gamma(k+1)}{n^k\Gamma(n^2+k)}.
\ee
In particular, 
\be
\mathbb{E}[\lambda_{n}^{\downarrow}]=1/n^3,\quad \text{and}\quad  \operatorname{Var}[\lambda_{n}^{\downarrow}]=\frac{1}{n^6}\left(\frac{n^2-1}{n^2+1}\right).
\ee

From the explicit formula~\eqref{eq:MBL} it is easy to see that the rescaled variable  $\lambda_n^{\downarrow}/n^3$ converges in distribution to an exponential variable with rate $1$:
\be
\lim_{n\to\infty}P\left(\lambda_{n}^{\downarrow}\geq x n^3\right)=e^{-x}.
\ee 
This is consistent with earlier results by Edelman \cite[Theorem 3.2]{Edelman88} on Wishart matrices (without fixed trace).

The distribution of $\lambda_n^{\downarrow}$  for $m>n$ does not seem to have been worked out to a similar degree of detail.  One can however infer asymptotic formulae  from known results on Wishart matrices without fixed trace. 
Let $(x_1,\dots,x_n)$ be distributed according to $G_{n,m}$ in~\eqref{eq:jointEigWishart}, where $m=cn$ with $c>1$. Set $a=(1-\sqrt{c})^2$ and $b=|1-\sqrt{c}|^{4/3}/c^{1/6}$. Then, the rescaled variable $bn^{-1/3}(x_{n}^{\downarrow}-an)$ converges in distribution to a $\beta=2$ Tracy-Widom random variable:
\be
\lim_{n\to\infty}P\left(\frac{x_{n}^{\downarrow}-na}{bn^{1/3 }}\leq x\right)=F_{2}(x).
\label{eq:TW}
\ee 
(See~\cite{perret2016finite} for more details). On the other hand $\sum_{k=1}^nx_k/cn^2$ converges in probability to $1$ by the law of large numbers. Hence, using the representation~\eqref{eq:rep_lambda_x}, we see that $\lambda_{n}^{\downarrow}$ has mean $\Or(n^{-1})$ and fluctuation $\Or(n^{-5/3})$ in the asymptotic regime $n\to\infty$.  


\end{document}